\documentclass{amsart}

\usepackage{latexsym}
\usepackage{linearb}
\usepackage[LGR, OT1]{fontenc}
\usepackage{amsmath,amsthm}
\usepackage{amssymb}
\usepackage{upgreek}
\usepackage[greek,english]{babel}
\usepackage{teubner}
\usepackage{MnSymbol, wasysym}
\usepackage{graphicx}
\usepackage{marvosym}
\usepackage[margin=1.25in,dvips]{geometry}
\usepackage{color}
\usepackage{comment}
\usepackage{mathrsfs}
\usepackage{mathtools}

\usepackage{bbm}
\usepackage{bbold}

\usepackage{slashed}

\usepackage[toc,page]{appendix}

\usepackage{physics}

\usepackage{amsaddr}

\usepackage{xcolor}
 
\usepackage{hyperref}
\hypersetup{
  colorlinks = true,
  linkcolor  = black
}

\newtheorem{innercustomTheorem}{Theorem}
\newenvironment{customTheorem}[1]
  {\renewcommand\theinnercustomTheorem{#1}\innercustomTheorem}
  {\endinnercustomTheorem}
\newtheorem{innercustomCorollary}{Corollary}
\newenvironment{customCorollary}[1]
  {\renewcommand\theinnercustomCorollary{#1}\innercustomCorollary}
  {\endinnercustomCorollary}
\def\f12{\frac 1 2}

\def\f12{\frac 1 2}

\makeindex

\numberwithin{equation}{section}

\newtheorem{definition}{Definition}[section]

\newtheorem{remark}{Remark}[section]
\newtheorem{lemma}{Lemma}[section]
\newtheorem{theorem}{Theorem}[section]
\newtheorem{proposition}{Proposition}[section]

\newtheorem{corollary}{Corollary}[section]

\newtheorem*{rough version}{Rough Version}

\newtheorem*{theorem*}{Theorem}

\newtheorem*{corollary*}{Corollary}

\newenvironment{sketch proof}{\proof}{\endproof}

\title[Morawetz estimates without relative degeneration]{Morawetz estimates without relative degeneration and exponential decay on Schwarzschild--de~Sitter spacetimes}

\author{Georgios Mavrogiannis}

\address{University of Cambridge, Department of Pure Mathematics and Mathematical Statistics, Wilberforce Road, Cambridge CB3 0WB, United Kingdom}

\email{gm615@cam.ac.uk}

\date\today

\begin{document}

\begin{abstract}
We use a novel physical space method to prove \emph{relatively} non-degenerate integrated energy estimates for the wave equation on subextremal Schwarzschild--de~Sitter spacetimes with parameters $(M,\Lambda)$. These are integrated decay statements whose bulk energy density, though degenerate at highest order, is everywhere comparable to the energy density of the boundary fluxes. As a corollary, we prove that solutions of the wave equation decay exponentially on the exterior region.

The main ingredients are a previous Morawetz estimate of Dafermos--Rodnianski and an additional argument based on commutation with a vector field which can be expressed in the form
\begin{equation*}
    r\sqrt{1-\frac{2M}{r}-\frac{\Lambda}{3}r^2}\frac{\partial}{\partial r},
\end{equation*}
where $\partial_r$ here denotes the coordinate vector field corresponding to a well chosen system of hyperboloidal coordinates.

Our argument gives exponential decay also for small first order perturbations of the wave operator. In the limit $\Lambda=0$, our commutation corresponds to the one introduced by Holzegel--Kauffman~\cite{gustav}.

\end{abstract}

\maketitle

\section{Introduction and motivation} 
Einstein's equation
\begin{equation}\label{einstein equation}
    Ric[g]-\Lambda g=0
\end{equation}
in the absence of matter with non-negative cosmological constant $\Lambda\geq 0$ has been extensively studied by both the mathematics and physics communities over the past century. Black hole solutions of~\eqref{einstein equation} are of particular interest. We will specifically here consider the Schwarzschild--de~Sitter spacetime $(\mathcal{M}_\textit{ext},g_{M,\Lambda})$ with
\begin{equation}\label{eq: prototype metric in SdS}
    g_{M,\Lambda}=-\left(1-\frac{2M}{r}-\frac{1}{3}\Lambda r^2\right)dt^2+\left(1-\frac{2M}{r}-\frac{1}{3}\Lambda r^2\right)^{-1}dr^2+r^2 d\sigma_{\mathbb{S}^2},
\end{equation}
where $d\sigma_{\mathbb{S}^2}$ is the standard metric of the unit sphere. This represents a black hole in an expanding cosmological universe, see~\cite{kottler,weyl,schwarzschild}. Making contact with a recent result of Holzegel--Kauffman~\cite{gustav}, we shall also consider the $\Lambda=0$ case, i.e.~the Schwarzschild spacetime $(\mathcal{M}_\textit{S},g_M)$ with $g_M=g_{M,\Lambda=0}$. The aim of this paper will be to revisit the study of the scalar wave equation 
\begin{equation}\label{eq: wave equation}
    \Box_g \psi =0,
\end{equation}
on the background~\eqref{eq: prototype metric in SdS}, studied in~\cite{bony,DR3,melr-barr-vasy,vasy1,vasy2,Dyatlov1}.

To give some motivation, we recall that over the years a number of methods have been developed to attack problems governed by linear and non-linear equations of hyperbolic type, including~\eqref{einstein equation}. A fundamental insight is the central role of generalizations of the energy concept as a tool for the global analysis of~\eqref{einstein equation}. To a large extent, this concept can be understood directly in physical space, i.e.~in the `time domain'. Physical space energy based methods have the advantage of displaying incredible resiliance when passing from statements concerning linear equations, e.g.~the scalar wave equation~\eqref{eq: wave equation}, to understanding non-linear problems, e.g.~Einstein's equation~\eqref{einstein equation}. A spectacular example of the success of this approach is the proof of the stability of Minkowski spacetime by Christodoulou and Klainerman, see the monograph~\cite{Chr-Klain}, and the more recent alternative proof~\cite{lindbladrodnianski}. Therefore, it is desirable, when possible, to have a physical space, purely energy based, understanding of boundedness and decay properties of~\eqref{eq: wave equation}. This is the goal of the present work. 

In the present paper we are specifically interested in the the exterior region of subextremal Schwarzschild--de~Sitter bounded between the event $\mathcal{H}^+$ and cosmological $\bar{\mathcal{H}}^+$ horizons, see already the dark shaded region of Figure~\ref{fig: extended penrose}. For an analysis of linear waves on the cosmological region see~\cite{Volker}, and analysis of linear waves on the black hole interior see~\cite{franzen,hintz1,grigoris}. The extremal case has not been studied systematically for $\Lambda>0$. It is subject to the Aretakis instability~\cite{aretakis}.

\begin{figure}[htbp]
    \centering
        \includegraphics[scale=0.6]{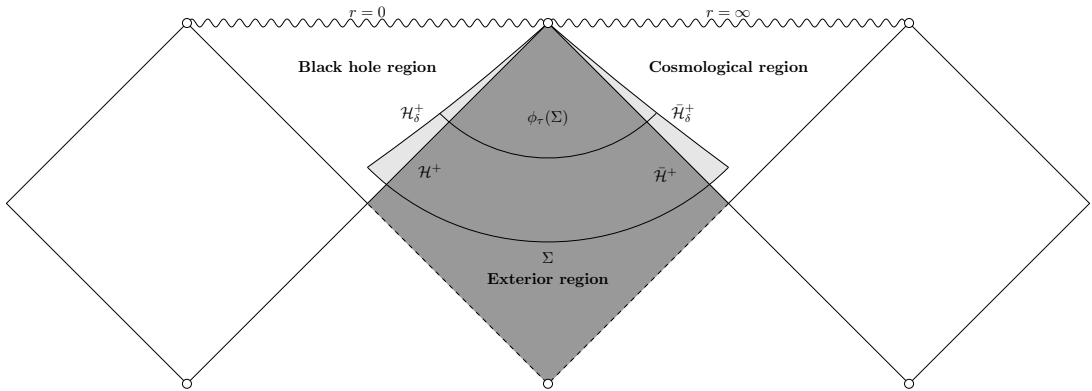}
        \caption{The Schwarzschild--de~Sitter spacetime}
        \label{fig: extended penrose}
\end{figure}

The wave equation~\eqref{eq: wave equation} on the subextremal Schwarzschild--de~Sitter exterior with $\Lambda>0$ has been studied in the past by two different, though related, approaches. 

One approach was initiated in~\cite{DR3} by Dafermos--Rodnianski, where the wave equation~\eqref{eq: wave equation} was studied using only energy estimates. Their energy estimates, which can in fact be expressed exclusively in physical space, prove faster than any polynomial decay in the shaded region of Figure~\ref{fig: extended penrose}, along a suitable foliation. Specifically, by assigning data on a spacelike hyperboloidal hypersurface
\begin{equation}
    \Sigma,
\end{equation}
that connects the event $\mathcal{H}^+$ and cosmological horizon $\bar{\mathcal{H}}^+$, they prove faster than any polynomial decay, in $\tau$, of the energy flux through the hypersurface
\begin{equation}\label{eq: prototype foliation}
    \phi_\tau(\Sigma),
\end{equation}
which is the push forward of $\Sigma$ by the Killing vector field $\partial_t$. This decay result, in turn, follows from a Morawetz estimate, which is of the form 
\begin{equation}\label{eq: prototype Morawetz estimate on Schwarzschild de Sitter}
    \int_{\tau_1}^{\tau_2} d\tau\int_{\phi_\tau(\Sigma)}(\psi-\psi_\infty)^2+ (\partial_{r}\psi)^2+ \left(1-\frac{3M}{r}\right)^2\left(|\slashed{\nabla}\psi|^2+(\partial_{\bar{t}}\psi)^2\right)\lesssim\int_{\phi_{\tau_1}(\Sigma)} (\partial_r\psi)^2+(\partial_{\bar{t}}\psi)^2+|\slashed{\nabla}\psi|^2,
\end{equation}
where $(r,\bar{t},\theta,\phi)$ are suitably defined hyperboloidal coordinates. Here, $\psi_\infty$ is a constant which can also be bounded from initial data. Note that the estimate~\eqref{eq: prototype Morawetz estimate on Schwarzschild de Sitter} manifestly already excludes finite frequency growing modes and, finally, is non-degenerate at the horizons $\mathcal{H}^+,\bar{\mathcal{H}}^+$, exploiting thus the red-shift. The estimate, however, degenerates at $r=3M$ due to the presence of trapped null geodesics, as necessitated by~\cite{ralston,sbierski}.

Another approach was initiated in~\cite{bony} by Bony--H\"afner, where they proved exponential decay for~\eqref{eq: wave equation}, restricted, however, away from the horizons $\mathcal{H}^+, \bar{\mathcal{H}}^+$, based on results concerning the asymptotic distribution of quasinormal modes shown previously by S\'a Barreto and Zworski in~\cite{barreto}. Following these proofs, a number of authors worked on the problem, see~\cite{melr-barr-vasy,vasy1}, and finally exponential decay was proved for the slowly rotating Kerr--de~Sitter black hole by Dyatlov in~\cite{Dyatlov1,Dyatlov2} without restriction away from the horizons. The papers of Dyatlov appeal to resolvent estimates in the complex plane and further machinery developed in microlocal analysis. Remarkably, building on all these results, Hintz and Vasy~\cite{hintz2} proved non-linear stability for the slowly rotating Kerr--de~Sitter spacetime. 

For nonlinear applications, arbitrarily fast polynomial decay is in fact more than sufficient, in principle, to obtain stability. Nonetheless, it is curious that the physical space argument of~\cite{DR3} only seemed to give this type of decay and not the full exponential decay. Indeed, this is connected precisely with the degeneration at $r=3M$ in the Morawetz estimate~\eqref{eq: prototype Morawetz estimate on Schwarzschild de Sitter}, referred to above.

The purpose of our paper is to overcome this difficulty and to show exponential decay for~\eqref{eq: wave equation} on Schwarzschild--de~Sitter, by an elementary additional physical space argument. We do so by proving a different type of local energy estimate, which, though still degenerate at $r=3M$, is \emph{relatively non-degenerate}, i.e.~its bulk term is not degenerate with respect to its boundary term, see already~\eqref{eq: main estimate of rough Theorem 1}. One ingredient of our proof is the Morawetz estimate~\eqref{eq: prototype Morawetz estimate on Schwarzschild de Sitter}. However, on top of the Morawetz estimate~\eqref{eq: prototype Morawetz estimate on Schwarzschild de Sitter}, we will require an additional commutation by a vector field which can again be thought to capture some of the properties of trapping, see already~\eqref{eq: prototype new vector fields}. In the $\Lambda=0$ case, this will recover a recent construction of Holzegel--Kauffman~\cite{gustav}. Our physical space commutation, in the high frequency limit, connects with the work of previous authors on `lossless estimates' and `non-trapping estimates', e.g. see~\cite{zworski,burq1,burq2,ikawa,nonnenmacher,hintz4,Dyatlov3,melr-barr-vasy}.

In our companion~\cite{mavrogiannis1}, we use the results of the present paper to prove global well posedness and exponential decay for the solutions of quasilinear wave equation and semilinear wave equations on Schwarzschild--de~Sitter.

Before stating our main results, we introduce the commutation and the energies which are of fundamental importance to this paper. 

\subsection{The commutation vector fields and the energy}
We use a system of regular hyperboloidal coordinates $(r,\bar{t},\theta,\phi)$, see already Section \ref{sec: preliminaries}, in which the metric takes the form 
\begin{equation}
    g_{M,\Lambda}=-\left(1-\frac{2M}{r}-\frac{\Lambda}{3}r^2\right)(d\bar{t})^2-2\frac{1-\frac{3M}{r}}{\sqrt{1-9M^2\Lambda}}\sqrt{1+\frac{6M}{r}} d\bar{t}dr +\frac{27M^2}{1-9M^2\Lambda}\frac{1}{r^2} (dr)^2+r^2 d\sigma_{\mathbb{S}^2}.
\end{equation}

The leaf 
\begin{equation}
\Sigma=\{\bar{t}=0\}    
\end{equation}
connects the event horizon $\mathcal{H}^+$ with the cosmological horizon $\bar{\mathcal{H}}^+$, as depicted in Figure~\ref{fig: G}.

We introduce the commutation vector field
\begin{equation}\label{eq: prototype new vector fields}
    \begin{aligned}
        \mathcal{G} = r\sqrt{1-\frac{2M}{r}-\frac{\Lambda}{3}r^2}\frac{\partial}{\partial r},
    \end{aligned}
\end{equation}
where $\partial_r$ is the coordinate vector field associated to $(r,\bar{t},\theta,\phi)$. (The simple form of this vector field is intimately related to the precise choice of coordinate $\bar{t}$. Note that $\mathcal{G}$ is orthogonal to the Killing vector field $\partial_{\bar{t}}$ at $r=3M$.)  

\begin{figure}[htbp]
    \begin{minipage}{.45\textwidth}
        \centering
        \includegraphics[scale=0.8]{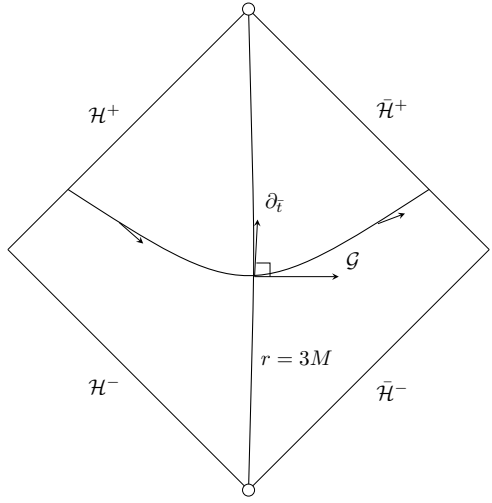}
        \caption{The vector field $\mathcal{G}$}
        \label{fig: G}
    \end{minipage}
\end{figure}

We define the energy density
\begin{equation}\label{eq: new energy}
    \begin{aligned}
        \mathcal{E}(\mathcal{G}\psi,\psi)\:\dot{=}\: \mathbb{T}(\partial_{\bar{t}},n)[\mathcal{G}\psi]+\mathbb{T}(n,n)[\psi],
    \end{aligned}
\end{equation}
where $\mathbb{T}$ is the energy momentum tensor of the wave equation, and $n$ is the normal to the foliation $\phi_\tau(\Sigma)$, see already Section~\ref{sec: preliminaries}.

The energy density \eqref{eq: new energy} is a non-negative definite quantity which contains up to second derivatives of $\psi$. We have the property 
\begin{equation}\label{eq: new energy similarity}
    \mathcal{E}(\mathcal{G}\psi,\psi)\sim (\partial_{\bar{t}}\mathcal{G}\psi)^2+\left(1-\frac{2M}{r}-\frac{\Lambda}{3}r^2\right)(\partial_r\mathcal{G}\psi)^2+|\slashed{\nabla}\mathcal{G}\psi|^2+(\partial_{\bar{t}}\psi)^2+(\partial_r\psi)^2+|\slashed{\nabla}\psi|^2. 
\end{equation}

In particular, we note that $\mathcal{E}(\mathcal{G}\psi,\psi)$ controls the $H^1$ energy density. Note, however, that $\mathcal{E}(\mathcal{G}\psi,\psi)$ does not control the full $H^2$ energy density.

\subsection{The main result}
Let 
\begin{equation}
    \phi_\tau(\Sigma)=\{\bar{t}=\tau\},
\end{equation}
for $\tau\geq 0$. We shall refer below to the energy density~\eqref{eq: new energy} and the commutation vector field~\eqref{eq: prototype new vector fields}. Our main theorem is an energy estimate \underline{without relative degeneration}.

\begin{customTheorem}{1}[rough version]\label{thm: Theorem 1 short version}
Solutions of the wave equation~\eqref{eq: wave equation} on the exterior of the subextremal Schwarzschild--de~Sitter $(\Lambda>0)$ black hole background $\left(\mathcal{M}_\textit{ext},g_{M,\Lambda}\right)$, dark shaded region of Figure~\ref{fig: extended penrose}, satisfy 
\begin{equation}\label{eq: main estimate of rough Theorem 1}
    \int_{\phi_{\tau_2}(\Sigma)}\mathcal{E}\left(\mathcal{G}\psi,\psi\right)+\int_{\tau_1}^{\tau_2} d\tau \int_{\phi_{\tau}(\Sigma)}\mathcal{E}(\mathcal{G}\psi,\psi)\lesssim \int_{\phi_{\tau_1}(\Sigma)} \mathcal{E}\left(\mathcal{G}\psi,\psi\right),
\end{equation}
for $0\leq\tau_1<\tau_2$.
\end{customTheorem}

\begin{remark}
Theorem~\ref{thm: Theorem 1 short version} remains true with $\Sigma$ replaced by a general spacelike hypersurface connecting the event with the cosmological horizons. 
\end{remark}

\begin{remark}
The estimate~\eqref{eq: main estimate of rough Theorem 1} differs from the Morawetz estimate~\eqref{eq: prototype Morawetz estimate on Schwarzschild de Sitter} in that, in the former, the same energy density appears in both the bulk term on the left hand side and the initial hypersurface flux term on the right hand side while in the latter, certain derivatives in the bulk term have $\left(1-\frac{3M}{r}\right)$ weights relative to their flux terms. In this sense, estimate~\eqref{eq: main estimate of rough Theorem 1} is \underline{relatively non-degenerate}. Note that this is still compatible with the obstructions of~\cite{ralston,sbierski} due to trapping at $r=3M$. It is this relative non-degeneracy that will allow us to immediately obtain exponential decay.
\end{remark}

As a corollary of Theorem~\ref{thm: Theorem 1 short version} we have exponential decay:

\begin{customCorollary}{1}[rough version]\label{cor: rough version exp decay}
With the assumptions of Theorem~\ref{thm: Theorem 1 short version}, we have the exponential decay estimates
\begin{equation*}
    \int_{\phi_{\tau}(\Sigma)}  \mathcal{E}(\mathcal{G}\psi,\psi) \lesssim e^{-c\tau}\int_{\Sigma}\mathcal{E}\left(\mathcal{G}\psi,\psi\right)
\end{equation*}
and 
\begin{equation*}
    \begin{aligned}
        &   \sup_{\phi_{\tau}(\Sigma)}|\psi-\psi_\infty|\lesssim \sqrt{E}e^{-c \tau},
    \end{aligned}
\end{equation*}
for $0\leq\tau$. Here $E$ is an appropriate higher order energy of the initial data of $\psi$, and $\psi_\infty$ is a constant that can be controlled by initial data.
\end{customCorollary}

We can apply the arguments of Theorem~\ref{thm: Theorem 1 short version} to obtain a second corollary, concerning small first order perturbations of the wave operator.

\begin{customCorollary}{2}[rough version]
Let $\psi$ be a solution of 
\begin{equation}
    \Box_{g_{M,\Lambda}}\psi=\epsilon a^j\partial_j\psi,
\end{equation}
where the vector field $a=a^j\partial_j$ is suitably bounded to first order, and $\epsilon$ sufficiently small. Then, the following estimate holds
\begin{equation*}
    \int_{\phi_{\tau}(\Sigma)}  \mathcal{E}(\mathcal{G}\psi,\psi)  \lesssim e^{-c\tau}\int_{\Sigma}\mathcal{E}\left(\mathcal{G}\psi,\psi\right).
\end{equation*}
Also, if $a=a^j\partial_j$ is suitably bounded up to second order, we obtain 
\begin{equation*}
    \begin{aligned}
        &   \sup_{\phi_{\tau}(\Sigma)}|\psi-\psi_\infty|\lesssim \sqrt{E} e^{-c \tau},
    \end{aligned}
\end{equation*}
for $0\leq\tau$, where $E$ is an appropriate integral quantity defined on $\Sigma$, and $\psi_\infty$ is a constant that can be controlled by initial data. 
\end{customCorollary}

In our forthcoming~\cite{mavrogiannis2}, we prove a Morawetz estimate on Kerr--de~Sitter spacetimes with parameters $(a,M,\Lambda)$ for the wave equation \eqref{eq: wave equation}, and more generally for the Klein--Gordon equation, and use it in conjunction with a generalization of the methods introduced here, to again prove an analogue of Theorem~\ref{thm: Theorem 1 short version} and exponential decay. We specifically establish exponential decay for slow rotation ($|a|\ll M,\Lambda$), or alternatively, in the full subextremal case of parameters but where the solution is assumed axisymmetric.

\subsection{The $\Lambda=0$ Schwarzschild limit}
In the Schwarzschild limit $\Lambda=0$ the commutation vector field $\mathcal{G}$~\eqref{eq: new vector fields} reduces to the vector field introduced in~\cite{gustav}, see already Section~\ref{sec: thm in Schwarzschild}. Note, however, that in~\cite{gustav} the commutation vector field $\mathcal{G}$ was expressed in coordinates that were there denoted as $(t,R^\star,\theta,\phi)$. Those coordinates, although regular at $\mathcal{H}^+$, do not coincide with our regular hyperboloidal coordinates. Thus, in those coordinates, $\mathcal{G}$ did not have the simple form~\eqref{eq: prototype new vector fields} for $\Lambda=0$.

\subsection{Acknowledgments}
The author thanks his supervisor M. Dafermos for numerous useful comments and helpful discussions. Moreover the author acknowledges the assistance of G. Moschidis, Y. Shlapentokh Rothman, D. Gajic and C. Kehle for various insightful conversations and important remarks.

\section{Preliminaries}\label{sec: preliminaries}

\subsection{The subextremality conditions}

We will use the following notation extensively. 

\begin{definition}
Let $M>0$ and $\Lambda\geq 0$. Then, we define the function $\mu$ by
\begin{equation}\label{eq: polynomial of the horizons}
    1-\mu_{M,\Lambda}=1-\frac{2M}{r}-\frac{1}{3}\Lambda r^2.
\end{equation}
We will often simply denote it as $\mu$.
\end{definition}

We define the following. 
\begin{definition}
Let $\Lambda\geq 0$. Then, the set of sub-extremal black hole parameters is
\begin{equation*}
    \begin{aligned}
        \mathcal{B}_\Lambda & =\{M>0:\:1-\mu=0\textit{ admits only distinct real roots}\}=\{M:0< M<\frac{1}{3\sqrt{\Lambda}}\}.
    \end{aligned}
\end{equation*}
For $M\in\mathcal{B}_\Lambda$, we denote the two positive real roots of $1-\mu$ as
\begin{equation}
r_+(M,\Lambda)<\bar{r}_+(M,\Lambda),    
\end{equation}
which will correspond to the area radius of the event and cosmological horizons respectively, see Definition~\ref{def: horizons}. We shall often denote these simply as $r_+,\bar{r}_+$. 
\end{definition}

\subsection{The spacetimes in regular hyperboloidal coordinates}\label{subsec: regular hyperboloidal corodinates}

We now define the metric of Schwarzschild--de~Sitter.

\begin{definition}\label{def: regular hyperboloidal coordinates}
Let $\Lambda > 0$ and $M\in\mathcal{B}_\Lambda$. We define the following manifold with boundary 
\begin{equation}
    \begin{aligned}
        \mathcal{M}_\textit{ext}=[r_+,\bar{r}_+]_r\times\mathbb{R}_{\bar{t}}\times\mathbb{S}^2_{(\theta,\phi)},\:\: \mathcal{M}_{\textit{ext},\delta}=[r_+-\delta,\bar{r}_+ +\delta]_r\times\mathbb{R}_{\bar{t}}\times\mathbb{S}^2_{(\theta,\phi)} 
    \end{aligned}
\end{equation}
for a $\delta>0$ sufficiently small, and the metric 
\begin{equation}\label{eq: regular metric de sitter}
    g_{M,\Lambda} =-(1-\mu)(d \bar{t})^2-2 \xi(r) d \bar{t} dr +\frac{1-\xi ^2(r)}{1-\mu}(dr)^2+r^2 d\sigma_{\mathbb{S}^2}
\end{equation}
where $d\sigma_{\mathbb{S}^2}$ is the standard metric of the unit sphere and 
\begin{equation}\label{eq: C(r)}
    \xi(r)= \frac{1-\frac{3M}{r}}{\sqrt{1-9M^2\Lambda}}\sqrt{1+\frac{6M}{r}},\quad \frac{1-\xi^2(r)}{1-\mu}=\frac{27M^2}{1-9M^2\Lambda}\frac{1}{r^2}.
\end{equation}
\end{definition}
We refer to the tuple $(r,\bar{t},\theta,\phi)$ as regular hyperboloidal coordinates. Note the inverse metric components
\begin{equation}
    \begin{aligned}
        g^{rr}=(1-\mu),\quad g^{\bar{t}\bar{t}}=-\frac{1-\xi^2(r)}{1-\mu},\quad g^{\bar{t}r}=-\xi(r).
    \end{aligned}
\end{equation}
See also Appendix \ref{sec: christoffel symbols} for the Christoffel symbols of the metric \eqref{eq: regular metric de sitter}.

\subsection{The time orientation}
We take the vector field
\begin{equation}
    \frac{\partial}{\partial \bar{t}}.
\end{equation}
to be future oriented. This defines a time orientation for $\mathcal{M}_\textit{ext}$.

\subsection{The event $\mathcal{H}^+$ and cosmological $\bar{\mathcal{H}}^+$ horizons}
Now we can define the following boundaries. 

\begin{definition}\label{def: horizons}
We define the following boundaries to the manifolds $\mathcal{M}_{\textit{ext}},\mathcal{M}_{\textit{ext},\delta}$
\begin{equation}
    \begin{aligned}
        &   \mathcal{H}^+=\{(r, \bar{t},\theta,\phi): \bar{t}>-\infty,r=r_+ \},\: \bar{\mathcal{H}}^+=\{(r, \bar{t},\theta,\phi): \bar{t}>-\infty,r=\bar{r}_+\},\\
        &   \mathcal{H}^+_\delta=\{(r, \bar{t},\theta,\phi): \bar{t}>-\infty,r=r_+-\delta \},\: \bar{\mathcal{H}}^+_\delta=\{(r, \bar{t},\theta,\phi): \bar{t}>-\infty,r=\bar{r}_++\delta\}
    \end{aligned}
\end{equation}
which we call future event horizon and future cosmological horizon respectively.
\end{definition}

\subsection{The photon sphere}
The hypersurface $\{r=3M\}$ is called the `photon sphere'. All future directed null geodesics either cross $\mathcal{H}^+$, cross $\bar{\mathcal{H}}^+$, or asymptote to $r=3M$. We will refer to the ones asymptoting to $r=3M$ as `future trapped null geodesics'.

\subsection{The Schwarzschild--de~Sitter coordinates}\label{subsec: the SdS and S coordinates}
We define the following coordinates. 
\begin{definition}\label{def: S,SdS coordinates}
From regular hyperboloidal coordinates, Definition~\ref{def: regular hyperboloidal coordinates}, we define the Schwarzschild--de~Sitter coordinates, $(r,t,\theta,\phi)\in (r_+,\bar{r}_+)\times \mathbb{R}\times \mathbb{S}^2$, by the following transformation 
\begin{equation}\label{eq: transformation of coordinates}
    t= \bar{t}+H(r),\quad H(r)=\int_{3M}^r\frac{\xi(\tilde{r})}{1-\mu}d\tilde{r},
\end{equation}
where $\xi(r)$ is given by~\eqref{eq: C(r)}. These coordinates cover the region $\mathcal{M}_{\textit{ext}}^o$.
\end{definition}

We rewrite the metric~\eqref{eq: regular metric de sitter} by using the transformation~\eqref{eq: transformation of coordinates} to obtain 
\begin{equation}\label{eq: metric in SdS}
    g_{M,\Lambda}=-\left(1-\mu\right)dt^2+\left(1-\mu\right)^{-1}dr^2+r^2d\sigma_{\mathbb{S}^2}.
\end{equation}

We distinguish between the coordinate vector field
\begin{equation}
    \boldsymbol{\frac{\partial}{\partial r}}
\end{equation}
with respect to the Schwarzschild--de~Sitter coordinate system $(r,t,\theta,\phi)$, and the coordinate vector field 
\begin{equation}
    \frac{\partial}{\partial r}
\end{equation}
with respect to the regular hyperboloidal coordinates $(r,\bar{t},\theta,\phi)$.

The following equality holds
\begin{equation}
\frac{\partial}{\partial \bar{t}}=\frac{\partial}{\partial t},     
\end{equation}
in the region $\mathcal{M}^o_\textit{ext}$ that the Schwarzschild--de~Sitter coordinates are defined.

\subsection{Tortoise coordinate}
We define the tortoise coordinate 
\begin{equation}\label{eq: tortoise coordinate r star}
    r^\star\:\dot{=}\:\int_{3M}^r\frac{1}{1-\mu} dr.
\end{equation}
The expression
\begin{equation}\label{eq: tortoise vector field}
    \frac{\partial}{\partial r^\star}
\end{equation}
will denote the vector field taken in $(r^\star,t,\theta,\phi)$ coordinates. Note that $\frac{\partial}{\partial r^\star}=(1-\mu)\boldsymbol{\frac{\partial}{\partial r}}$. The vector field \eqref{eq: tortoise vector field} extends smoothly to a tangential vector field along $\mathcal{H}^+$ and $\bar{\mathcal{H}}^+$.

\subsection{Chain rule between coordinate vector fields}\label{subsec: chain rule between coordinates}
By a simple chain rule, we have
\begin{equation}
    \boldsymbol{\frac{\partial}{\partial r}}=\frac{\partial}{\partial r}+\frac{d\bar{t}}{d r}\frac{\partial}{\partial\bar{t}}.
\end{equation}
We have
\begin{equation*}
    \frac{d\bar{t}}{dr}=-\frac{dH}{dr}(r),
\end{equation*}
where, for $H(r)$, see Section~\ref{subsec: the SdS and S coordinates}. Note that at $r=3M$ the vector field $\partial_r$ is in the direction of $\partial_{r^\star}$, since
\begin{equation}\label{eq: r star and r are the same on trapping}
    \frac{\partial}{\partial r}=\sqrt{\frac{27M^2}{1-9M^2\Lambda}}\frac{\partial}{\partial r^\star}.
\end{equation}

\subsection{The vector field $\mathcal{G}$ in Schwarzschild--de~Sitter coordinates}
We define the vector field
\begin{equation}\label{eq: new vector fields}
    \mathcal{G}=r\sqrt{1-\mu}\frac{\partial}{\partial r}.
\end{equation}
Note that $\mathcal{G}$ is $C^0$ on the horizons $\mathcal{H}^+,\bar{\mathcal{H}}^+$, but not $C^1$. We extend the vector field $\mathcal{G}$, of~\eqref{eq: new vector fields}, beyond the horizons $\mathcal{H}^+,\bar{\mathcal{H}}^+$ such that
\begin{equation}
	\mathcal{G}\equiv 0,\quad \text{in } \{r\leq r_+\}\cup \{r\geq \bar{r}_+\}.
\end{equation}
This vector field is suggested by the good commutation property of Proposition \ref{prop: commuted equation for G psi}. Note that in view of \eqref{eq: r star and r are the same on trapping}, at $r=3M$ the vector field $\mathcal{G}$ is in the direction of $\partial_{r^\star}$ and is thus orthogonal to $\partial_{\bar{t}}$. This is significant because this is precisely the derivative that does not degenerate in estimate \eqref{eq: prototype Morawetz estimate on Schwarzschild de Sitter}. 

We express the vector field~\eqref{eq: new vector fields} in Schwarzschild--de~Sitter coordinates.
\begin{lemma}
In $\mathcal{M}^o_{\textit{ext}}$ the vector field~\eqref{eq: new vector fields} can be written as
\begin{equation}
    \mathcal{G}=G_1(r)\frac{\partial}{\partial r^\star}+G_2(r)\frac{\partial}{\partial t},
\end{equation}
in Schwarzschild--de~Sitter coordinates. The functions $G_1,G_2$ are defined as follows 
\begin{equation*}
    G_1(r)=\sqrt{\frac{r^2}{1-\mu}},\:\:G_2(r)=\frac{r}{\sqrt{1-\mu}}\left(1-\frac{3M}{r}\right)\frac{1}{\sqrt{1-9M^2\Lambda}}\sqrt{1+\frac{6M}{r}}.
\end{equation*}
\end{lemma}
\begin{proof}
We note
\begin{equation}
    \begin{aligned}
        \left(G_1\frac{\partial}{\partial r^\star}+G_2\frac{\partial}{\partial t}\right) &= r\sqrt{1-\mu}\frac{\partial}{\partial r}+\frac{r}{\sqrt{1-\mu}}\left(-(1-\mu)\frac{d H}{d r}+\frac{(1-\frac{3M}{r})}{\sqrt{1-9M^2\Lambda}}\sqrt{1+\frac{6M}{r}}\right)\frac{\partial}{\partial t}= r\sqrt{1-\mu}\frac{\partial}{\partial r},
    \end{aligned}
\end{equation}
since $(1-\mu)\frac{d H}{d r}=\frac{1-\frac{3M}{r}}{\sqrt{1-9M^2\Lambda}}\sqrt{1+\frac{6M}{r}}$. 
\end{proof}

\subsection{Wave operator}
We denote by $\slashed{\nabla}$ the covariant derivative with respect to $r^2 d\sigma_{\mathbb{S}^2}$.

The wave operator is
\begin{equation}\label{eq: wave operator}
    \begin{aligned}
        \Box_{g_{M,\Lambda}} \psi = -\frac{1-\xi^2(r)}{1-\mu}\partial_t^2\psi+(1-\mu)\partial_r^2\psi &+2\xi(r)\partial_r\partial_t\psi+\slashed{\nabla}^A\slashed{\nabla}_A \psi \\
        &   + \left(-\partial_r\xi +2r^{-1}\xi(r)\right)\partial_{\bar{t}}\psi+\left(2r^{-1}(1-\mu)+\partial_r(1-\mu)\right)\partial_r\psi
    \end{aligned}
\end{equation}
where $\slashed{\Delta}=\slashed{\nabla}^A\slashed{\nabla}_A$. Note, moreover, that we define the following expression
\begin{equation}
    |\slashed{\nabla}\psi|^2=\slashed{\nabla}^A\psi\slashed{\nabla}_A\psi.
\end{equation}

Furthermore, we denote as $\Omega_\alpha$, for $\alpha=1,2,3$ the standard vector fields
\begin{equation}\label{eq: generators of lie algebra}
    \Omega_1=\partial_\phi,\quad\Omega_2=\cos{\phi}\:\partial_\theta-\sin{\phi}\cot{\theta}\:\partial_\phi,\quad \Omega_3=-\left(\sin{\phi}\:\partial_{\theta}+\cos{\phi}\cot{\theta}\:\partial_\phi\right),
\end{equation}
that generate the lie algebra $so(3)$. Note $[\Box_{g_{M,\Lambda}},\Omega_\alpha]=0, [\Box_{g_{M,\Lambda}},\partial_{\bar{t}}]=0$.

\subsection{Spacelike hypersurfaces}\label{subsec: acausal hypersurfaces, in SdS}
In Schwarzschild--de~Sitter a prototype hypersurface that is spacelike and connects the event horizon $\mathcal{H}^+$ with the cosmological horizon $\bar{\mathcal{H}}^+$ would be
\begin{equation}
    \Sigma=\{\bar{t}=0\}.
\end{equation}

Our results also hold for general spacelike hypersurfaces connecting $\mathcal{H}^+$ and $\bar{\mathcal{H}}^+$. (For convenience, however, we always work with $\Sigma$ fixed as above.)

\subsection{Spacelike foliations and causal domains}\label{subsec: causal domains}
We push forward the hypersurface $\Sigma$, see Section~\ref{subsec: acausal hypersurfaces, in SdS}, under the flow $\phi_\tau$ of the vector field $\partial_{\bar{t}}$ to obtain the family of hypersurfaces
\begin{equation}
\phi_\tau(\Sigma)=\{\bar{t}=\tau\}.
\end{equation}

We define the following spacetime domains.
\begin{definition}\label{def: spacetime domain}
For $\tau_1<\tau_2$ define the spacetime domain
\begin{equation}
    D(\tau_1,\tau_2)\:\dot{=}\:J^+(\phi_{\tau_1}(\Sigma))\cap J^-(\phi_{\tau_2}(\Sigma))=\{(r,\bar{t},\theta,\phi):\tau_1\leq \bar{t}\leq \tau_2\},
\end{equation}
and 
\begin{equation}
    D(\tau,\infty)\:\dot{=}\:J^+(\phi_\tau(\Sigma))=\{(r,\bar{t},\theta,\phi):\tau\leq \bar{t}<\infty\}.
\end{equation}
\end{definition}

The domains of Definition \ref{def: spacetime domain} are both globally hyperbolic, with  $\phi_{\tau_1}(\Sigma)$ and $\phi_{\tau}(\Sigma)$, respectively, as Cauchy hypersurfaces. 

\subsection{Normals of spacelike hypersurfaces}\label{subsec: spacelike normals}
The unit normal vector fields of the foliation of Section~\ref{subsec: acausal hypersurfaces, in SdS} can be computed from the gradient $\nabla\bar{t}$ which in regular hyperboloidal coordinates is 
\begin{equation}
    \nabla\bar{t}=-\frac{1-\xi^2(r)}{1-\mu}\frac{\partial}{\partial \bar{t}}-\xi(r)\frac{\partial}{\partial r}=- \frac{27M^2}{1-9M^2\Lambda}\frac{1}{r^2}\frac{\partial}{\partial \bar{t}}-\frac{1-\frac{3M}{r}}{\sqrt{1-9M^2\Lambda}}\sqrt{1+\frac{6M}{r}}\frac{\partial}{\partial r}.
\end{equation}
Note moreover that $g(\nabla\bar{t},\nabla\bar{t})\leq b(M,\Lambda)<0$ on $[r_+,\bar{r}_+]$. The normal of the relevant foliation will be 
\begin{equation}
    n_{\phi_{\tau}(\Sigma)}=-\frac{\nabla \bar{t}}{\sqrt{-g(\nabla \bar{t},\nabla\bar{t})}}=\sqrt{\frac{1-\xi^2}{1-\mu}}\frac{\partial}{\partial \bar{t}}+\frac{\xi(r)}{\sqrt{\frac{1-\xi^2}{1-\mu}}}\frac{\partial}{\partial r},
\end{equation}
often denoted simply as $n$.

\subsection{Volume forms}\label{subsec: volume forms}
The volume form of a spacetime domain is 
\begin{equation}\label{eq: hyperb volume form}
    dg=r^2\sin\theta d\bar{t} dr d\theta d\phi
\end{equation}
with respect to the $\left(r,\bar{t},\theta,\phi\right)$ coordinates.

By pulling back the spacetime volume form~\eqref{eq: hyperb volume form} into hypersurfaces of constant $\bar{t}$, we obtain that the $\{\bar{t}=\tau\}$ hypersurfaces admit the volume form 
\begin{equation}\label{eq: volume form of spacelike hypersurface}
dg_{\{\bar{t}=c\}}=r\sqrt{\frac{27M^2}{1-9M^2\Lambda}}\sin\theta dr d\theta d\phi.    
\end{equation}

We define the normals of the event and cosmological horizons respectively as 
\begin{equation}
    n_{\mathcal{H}^+}=\partial_{\bar{t}},\:n_{\bar{\mathcal{H}}^+}=\partial_{\bar{t}}.
\end{equation}

With the above choice of normals, the corresponding volume forms of the respective null hypersurfaces take the form
\begin{equation*} 
dg_{\mathcal{H}^+}=r^2\sin\theta d\bar{t}d\sigma_{\mathbb{S}^2},\:\: dg_{\bar{\mathcal{H}}^+}=r^2\sin\theta d\bar{t}d\sigma_{\mathbb{S}^2}.
\end{equation*}

In all integrals without explicit volume form, it is to be understood that the volume forms are taken to be the ones defined in this Section.

\subsection{Coarea formula}\label{subsec: coarea formula}

Let $f$ be a continuous non-negative function. Then, note the coarea formula 
\begin{equation}\label{eq: subsec: coarea formula, eq 1}
\int_{0}^\tau d\tau\int_{\phi_\tau(\Sigma)} f dg_{\{\bar{t}=\tau\}}\sim\int\int_{D(0,\tau)}\frac{1}{r} f dg,
\end{equation}
where the constants in the above similarity depend only on the black hole mass $M$ and do not degenerate in the limit $\Lambda\rightarrow 0$. For fixed $\Lambda>0$ the $r$ factor of~\eqref{eq: subsec: coarea formula, eq 1} is of course inessential.

\subsection{Penrose diagrams}\label{subsec: penrose diag}
The reader familiar with the Penrose diagrammatic representation may wish to refer to Figure~\ref{fig: Penrose diagram}.
\begin{figure}[htbp]
    \begin{minipage}{.45\textwidth}
        \centering
        \includegraphics[scale=0.8]{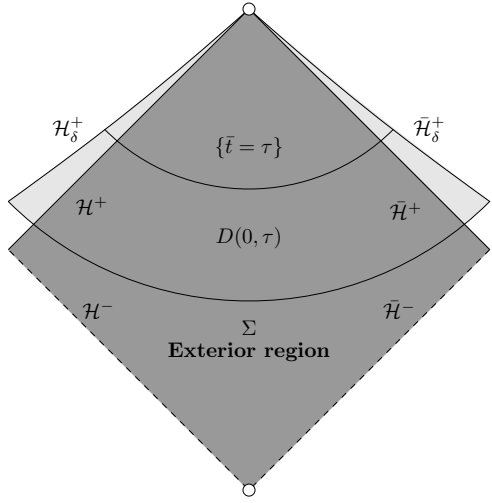}
        \caption{The foliation of the Schwarzschild--de~Sitter exterior}
        \label{fig: Penrose diagram}
    \end{minipage}
\end{figure}

\subsection{Currents and the divergence theorem}\label{subsec: currents, div theorem}
We will employ the energy momentum tensor and the relevant current it produces. 

\begin{definition}\label{def: energy momentum tensor}
Let $g$ be a smooth Lorentzian metric. For $\psi$ a solution of 
\begin{equation}\label{eq: inhomogeneous wave equation}
\Box_g\psi=F,    
\end{equation}
we define the energy momentum tensor
\begin{equation}\label{eq: energy momentum tensor}
    \mathbb{T}_{\mu\nu}[\psi]\:\dot{=}\:\partial_{\mu}\psi\partial_{\nu}\psi-\frac{1}{2}g_{\mu\nu}\Big(g^{\gamma\epsilon}\partial_{\gamma}\psi\partial_{\epsilon}\psi  \Big).
\end{equation}
The energy current with respect to a vector field $X$ is 
\begin{equation}\label{eq: current}
    J^{X}_{\mu}[\psi]\:\dot{=}\: \mathbb{T}_{\mu\nu}[\psi]X^{\nu},
\end{equation}
with divergence
\begin{equation}
    \nabla^{\mu}J^{X}_{\mu} =X(\psi) F+\frac{1}{2}\mathbb{T}_{\mu\nu}\:^{(X)}\pi^{\mu\nu},
\end{equation}
where 
\begin{equation}
    ^{(X)}\pi^{\mu\nu}\dot{=}\:\frac{1}{2}\big( \nabla^{\mu}X^{\nu}+\nabla^{\nu}X^{\mu}  \big)
\end{equation}
is the deformation tensor. 

Lastly, for $X,n$ future causal vector fields, we have that
\begin{equation}\label{eq: nonnegative energy momentum tensor}
    J^X_\mu[\psi]n^\mu=\mathbb{T}(X,n)[\psi]\geq 0
\end{equation}
\end{definition}

We apply the divergence theorem in the region $D(\tau_1,\tau_2)$ to obtain the following Proposition. 

\begin{proposition}\label{prop: divergence Theorem}
Let $\psi$ satisfy the equation~\eqref{eq: inhomogeneous wave equation} on $D(\tau_1,\tau_2)$. Then, with the notation above, the following holds
\begin{equation}\label{eq: divergence Theorem}
    \begin{aligned}
        \int_{\phi_{\tau_2}(\Sigma)}J^{X}_{\mu}[\psi]n^{\mu}_{\phi_{\tau_2}(\Sigma)}+\int_{\mathcal{H}^{+}\cap D(\tau_1,\tau_2)}J^{X}_{\mu}[\psi]n^{\mu}_{\mathcal{H}^{+}}+\int_{\bar{\mathcal{H}}^{+}\cap D(\tau_1,\tau_2)}J^{X}_{\mu}[\psi]n^{\mu}_{\bar{\mathcal{H}}^{+}}& +\int\int_{D(\tau_1,\tau_2)}\nabla^{\mu}J^{X}_{\mu}[\psi] \\
        &   =\int_{\phi_{\tau_1}(\Sigma)} J^{X}_{\mu}[\psi]n^{\mu}.
    \end{aligned}
\end{equation}
\end{proposition}
In the above, in accordance with our conventions from Section~\ref{subsec: volume forms}, all integrals are taken with respect to the volume form of the respective hypersurfaces or spacetime domain.

For a further study on currents related to partial differential equations, see the monograph of Christodoulou~\cite{christodoulou}. 

\section{The Main theorems}\label{sec: main theorems}

\subsection{The homogeneous wave equation}
The following is the central theorem of this paper.

\begin{customTheorem}{1}[detailed version]\label{main theorem 1}
We fix the parameters $M,\Lambda>0$. Then, there exists a constant $C=C(M,\Lambda)>0$, such that for $\psi$ a sufficiently regular solution of the wave equation~\eqref{eq: wave equation} on $D(\tau_1,\tau_2)$, we have
\begin{equation}\label{eq: estimate of thm 1, 1}
    \int_{\phi_{\tau_2}(\Sigma)}\mathcal{E}\left(\mathcal{G}\psi,\psi\right)+\int\int_{D(\tau_1,\tau_2)}\left(\frac{1}{(1-\mu)}(\partial_t\mathcal{G}\psi)^2+\mathcal{E}(\mathcal{G}\psi,\psi)\right)\leq C \int_{\phi_{\tau_1}(\Sigma)} \mathcal{E}\left(\mathcal{G}\psi,\psi\right),
\end{equation}
for all $0\leq\tau_1<\tau_2$, where $\mathcal{E}\left(\mathcal{G}\psi,\psi\right)$ is defined in~\eqref{eq: new energy}. 

Since $\int\int_{D(\tau_1,\tau_2)} \mathcal{E}(\mathcal{G}\psi,\psi)\sim \int_{\tau_1}^{\tau_2}d\tau\int_{\phi_\tau(\Sigma)}\mathcal{E}(\mathcal{G}\psi,\psi)$ we in particular have 
\begin{equation}\label{eq: estimate of thm 1, 2}
    \int_{\phi_{\tau_2}(\Sigma)}\mathcal{E}\left(\mathcal{G}\psi,\psi\right)+\int_{\tau_1}^{\tau_2}d\tau\int_{\phi_\tau(\Sigma)}\mathcal{E}(\mathcal{G}\psi,\psi)\leq C \int_{\phi_{\tau_1}(\Sigma)} \mathcal{E}\left(\mathcal{G}\psi,\psi\right).
\end{equation}

Note that this Theorem also holds for the domain $D_\delta(\tau_1,\tau_2)$ in the place of $D(\tau_1,\tau_2)$, for a sufficiently small $\delta>0$, where the constant $C$ is independent of $\delta$.
\end{customTheorem}

As an immediate corollary we have exponential decay.

\begin{customCorollary}{1}[detailed version]\label{corollary on exponential decay}
There exist positive constants $c(M,\Lambda)>0, C(M,\Lambda)>0$, such that for $\psi$ a sufficiently regular solution of the wave equation~\eqref{eq: wave equation} on $D(0,\tau)$, we have
\begin{equation*}
    \begin{aligned}
        \int_{\phi_{\tau}(\Sigma)}\mathcal{E}(\mathcal{G}\psi,\psi) \leq C e^{-c\tau} \int_{\Sigma} \mathcal{E}\left(\mathcal{G}\psi,\psi\right),
    \end{aligned}
\end{equation*}
and also the pointwise decay 
\begin{equation*}
    \begin{aligned}
        &   \sup_{\phi_{\tau}(\Sigma)}|\psi-\psi_\infty|\leq C\sqrt{E}e^{-c \tau},
    \end{aligned}
\end{equation*}
where $\psi_\infty$ is a constant satisfying $|\psi_\infty|\leq C\left( \sup_{\Sigma}|\psi|+ \sqrt{\int_{\Sigma}\mathbb{T}(n,n)[\psi]}\right)$, and 
\begin{equation}
    E=\sum_{0\leq i\leq 1}\sum_{1\leq \alpha\leq 3}\int_\Sigma  \mathcal{E}(\mathcal{G}\Omega_\alpha^i\psi,\Omega_\alpha^i\psi),
\end{equation}
where $\Omega_\alpha$ are defined in \eqref{eq: generators of lie algebra}. 
\end{customCorollary}

\subsection{The inhomogeneous wave equation and absorption of small error terms}
The following theorem concerns the inhomogeneous wave equation on the Schwarzschild--de~Sitter background.

\begin{customTheorem}{2}[detailed version]\label{thm: absorption of small terms}
Let F be a sufficiently regular function on $D(\tau_1,\tau_2)$. We have that, for sufficiently regular solutions of 
\begin{equation}\label{eq: wave equation with error terms}
    \Box_{g_{M,\Lambda}}\psi=F
\end{equation}
on $D(\tau_1,\tau_2)$, the following holds
\begin{equation}\label{eq: absorption of small terms}
    \begin{aligned}
        \int_{\phi_{\tau_2}(\Sigma)}\mathcal{E}(\mathcal{G}\psi,\psi)&+\int\int_{D(\tau_1,\tau_2)}\left(\frac{1}{(1-\mu)}(\partial_t\mathcal{G}\psi)^2+\mathcal{E}(\mathcal{G}\psi,\psi)\right)\\
        &   \leq C\int_{\phi_{\tau_1}(\Sigma)}\mathcal{E}(\mathcal{G}\psi,\psi)+C\int\int_{D(\tau_1,\tau_2)}(1-\mu)\left|\mathcal{G}F\right|^2+|F|^2,
    \end{aligned}
\end{equation}
for all $0\leq\tau_1<\tau_2$, and for some constant $C(M,\Lambda)$.

Note that this Theorem also holds for the domain $D_\delta(\tau_1,\tau_2)$ in the place of $D(\tau_1,\tau_2)$ for a sufficiently small $\delta>0$, where the constant $C$ is independent of $\delta$. 
\end{customTheorem}

We have the following Corollary.

\begin{customCorollary}{2}[detailed version]\label{cor: absorption of small terms}
Let $a=a^j\partial_j$ be a vector field, where
\begin{equation}
a^{\bar{t}}, a^r, g^{\theta\theta}(a^\theta)^2+g^{\phi\phi}(a^\phi)^2, [\mathcal{G},a]^{\bar{t}}, [\mathcal{G},a]^r, g^{\theta\theta}([\mathcal{G},a]^{\theta})^2+g^{\phi\phi}([\mathcal{G},a]^\phi)^2    
\end{equation}
are smooth and bounded functions on $D(0,\infty)$. For $\epsilon>0$ sufficiently small, there exists a constant $C(M,\Lambda)>0$ such that sufficiently regular solutions of the equation 
\begin{equation}\label{eq: small error equation SdS}
\Box_{g_{M,\Lambda}}\psi=\epsilon a^j\partial_j\psi
\end{equation}
satisfy the following estimate
\begin{equation}
    \int_{\phi_{\tau_2}(\Sigma)}\mathcal{E}\left(\mathcal{G}\psi,\psi\right) +\int\int_{D(\tau_1,\tau_2)}\left(\frac{1}{(1-\mu)}(\partial_t\mathcal{G}\psi)^2+\mathcal{E}(\mathcal{G}\psi,\psi)\right)\leq C\int_{\phi_{\tau_1}(\Sigma)}\mathcal{E}\left(\mathcal{G}\psi,\psi\right).
\end{equation}
Also, there exist constants $C(M,\Lambda)>0$, $c(M,\Lambda)>0$ depending only on the black hole parameters such that 
\begin{equation}
    \int_{\phi_{\tau}(\Sigma)}\mathcal{E}(\mathcal{G}\psi,\psi)\leq C e^{-c \tau} \int_{\Sigma} \mathcal{E}(\mathcal{G}\psi,\psi).
\end{equation}
Moreover, let $[\Omega_\alpha\mathcal{G},a]^{\bar{t}}, [\Omega_\alpha\mathcal{G},a]^r,  g^{\theta\theta}([\Omega_\alpha\mathcal{G},a]^\theta)^2+g^{\phi\phi}([\Omega_\alpha\mathcal{G},a]^\phi)^2 $ be bounded for all $\alpha$, where $\Omega_\alpha$ are defined in equation \eqref{eq: generators of lie algebra}. Then, we obtain 
\begin{equation}
    \sup_{\phi_{\tau}(\Sigma)}|\psi-\psi_\infty|\leq C\sqrt{E}e^{-c\tau},
\end{equation}
where $\psi_\infty$ is a constant satisfying $|\psi_\infty|\leq C\left( \sup_{\Sigma}|\psi|+ \sqrt{\int_{\Sigma}\mathbb{T}(n,n)[\psi]+\mathbb{T}(n,n)[\mathcal{G}\psi]}\right)$, and 
\begin{equation}
    E=\sum_{0\leq i\leq 1}\sum_{1\leq\alpha\leq 3} \int_{\Sigma} \mathcal{E}(\mathcal{G}\Omega_\alpha^i \psi,\Omega_\alpha^i\psi).
\end{equation}
\end{customCorollary}

\subsection{The higher order statement}\label{subsec: thm absorption of small terms}

The following Theorem is the higher order statement of Theorem \ref{thm: absorption of small terms}. We will use the following result in our companion paper~\cite{mavrogiannis1} to prove stability of solutions of the quasilinear wave equation. 

\begin{customTheorem}{3}\label{thm: higher order G estimate}
Let $F$ be a sufficiently regular function on $D(\tau_1,\tau_2)$. We have that, for a sufficiently regular solution of 
\begin{equation}\label{eq: cor: higher order G estimate, eq 0}
    \Box_{g_{M,\Lambda}}\psi=F
\end{equation}
on $D(\tau_1,\tau_2)$, and for any $j\geq 3$, there exists a constant $C=C(j,M,\Lambda)$ such that the following higher order estimate holds
\begin{equation}\label{eq: cor: higher order G estimate, eq 1}
    \begin{aligned}
        &	E_{\mathcal{G},j}[\psi](\tau_2) \\   &	\quad\quad +\int\int_{D(\tau_1,\tau_2)}        \frac{1}{1-\mu}\sum_{0\leq i_1+i_2\leq j-2}\sum_\alpha\left(\partial_{\bar{t}}^{1+i_1}\Omega_\alpha^{i_2}\mathcal{G}\psi\right)^2+\int_{\tau_1}^{\tau_2}d\tau E_{\mathcal{G},j}[\psi](\tau)\\
        &   \quad \leq C  E_{\mathcal{G},j}[\psi](\tau_1)+C\int\int_{D(\tau_1,\tau_2)} (1-\mu)\sum_{0\leq i_1+i_2\leq j-2}\sum_\alpha\left(\partial_{\bar{t}}^{i_1}\Omega_\alpha^{i_2}\mathcal{G}F\right)^2+\sum_{0\leq i_1+i_2+i_3\leq j-2}\sum_\alpha\left(\partial_{\bar{t}}^{i_1}\Omega_\alpha^{i_2}\partial_r^{i_3}F\right)^2\\
        &	\quad\quad\quad\quad\quad\quad\quad +C \int_{\{\bar{t}=\tau_2\}} \sum_{0\leq i_1+i_2+i_3\leq j-3}\sum_\alpha(1-\mu)^{2i_3+1}\left(\partial_{\bar{t}}^{i_1}\Omega_\alpha^{i_2}\partial_r^{i_3}\mathcal{G}F\right)^2+\sum_{0\leq i_1+i_2+i_3\leq j-3}\sum_\alpha\left(\partial_{\bar{t}}^{i_1}\Omega_\alpha^{i_2}\partial_r^{i_3}F\right)^2
    \end{aligned}
\end{equation}
for all $0\leq \tau_1\leq \tau_2$, where
\begin{equation}\label{eq: cor: higher order G estimate, eq 2}
\begin{aligned}
E_{\mathcal{G},j}[\psi](\tau) &   =\sum_\alpha\int_{\{\bar{t}=\tau\}}\sum_{1\leq i_1+i_2+i_3\leq j-1,i_3\geq 1}(1-\mu)^{2i_3-1}\left(\partial_{\bar{t}}^{i_1}\Omega_\alpha^{i_2}\partial_r^{i_3}\mathcal{G}\psi\right)^2+\sum_{1\leq i_1+i_2\leq j-1}\left(\partial_{\bar{t}}^{i_1}\Omega_\alpha^{i_2}\mathcal{G}\psi\right)^2 \\
&	\quad\quad\quad\quad\quad\quad\quad+\sum_{1\leq i_1+i_2+i_3\leq j-1}\left(\partial_{\bar{t}}^{i_1}\Omega_\alpha^{i_2}\partial_r^{i_3}\psi\right)^2.
\end{aligned}
\end{equation} 
The $j=2$ case is the same without the hypersurface error terms in the right hand side of \eqref{eq: cor: higher order G estimate, eq 1}.

Note that this Theorem also holds for the domain $D_\delta(\tau_1,\tau_2)$ in the place of $D(\tau_1,\tau_2)$, for a sufficiently small $\delta>0$, where the constant $C$ is independent of $\delta$.
\end{customTheorem}
\begin{proof}
	See Section \ref{sec: higher order G estimate}.
\end{proof}

\section{Proof of Theorem~\ref{main theorem 1}}\label{sec: proof of main theorem 1}

\subsection{The Morawetz estimate of~\cite{DR3}.}
As mentioned earlier, this proof utilizes a Morawetz estimate for the wave equation on Schwarzschild--de~Sitter, which we find in~\cite{DR3}. 
\begin{theorem}[Theorem 1.1 in~\cite{DR3}]\label{thm: integrated local energy decay for psi}
There exists a constant $B=B(\Sigma,M,\Lambda)>0$ such that, for $\psi$ satisfying the wave equation~\eqref{eq: wave equation} in $D(\tau_1,\tau_2)$
\begin{equation}\label{eq: integrated local energy decay for psi}
    \int\int_{D(\tau_1,\tau_2)} (\psi-\psi_\infty)^2+(\partial_{r}\psi)^2 +\left(1-\frac{3M}{r}\right)^2\left((\partial_{\bar{t}}\psi)^2+\left|\slashed{\nabla}\psi\right|^2\right)\leq B \int_{\phi_{\tau_1}(\Sigma)} J_\mu^n[\psi]n^\mu,
\end{equation}
and
\begin{equation}\label{eq: boundedness estimate}
    \left(\int_{\phi_{\tau_2}(\Sigma)}+\int_{\mathcal{H}^+\cap D(\tau_1,\tau_2)}+\int_{\bar{\mathcal{H}}^+\cap D(\tau_1,\tau_2)}\right)J^n_\mu[\psi] n^\mu\leq B\int_{\phi_{\tau_1}(\Sigma)}J^n_\mu[\psi] n^\mu
\end{equation}
where the constant $\psi_\infty$ satisfies $|\psi_\infty|\leq C\left( \sup_{\phi_{\tau_1}(\Sigma)}|\psi|+ \sqrt{\int_{\phi_{\tau_1}(\Sigma)}J^n_\mu[\psi] n^\mu}\right)$. 
\end{theorem}

\begin{remark}
At $r=3M$, corresponding to trapping, note that the vector field $\partial_r$ does not degenerate and is, in fact, in the direction of $\partial_{r^\star}$ by \eqref{eq: r star and r are the same on trapping}. We have used this relation to translate Theorem~\ref{thm: integrated local energy decay for psi} from the coordinates of~\cite{DR3}.
\end{remark}

\begin{remark}
The proof of Theorem~\ref{thm: integrated local energy decay for psi} of Dafermos--Rodnianski actually made use of a spherical harmonics decomposition. It would be interesting to give a completely physical space proof, as has been done for Schwarzschild, see~\cite{DR7}.
\end{remark}

\begin{remark}\label{rem: on the inhomogeneous Morawetz}
For a solution of the inhomogeneous equation 
\begin{equation}
    \Box_{g_{M,\Lambda}}\psi=F,
\end{equation}
one sees that the same estimate applies, with an additional term 
\begin{equation}
    \int\int_{D(\tau_1,\tau_2)}|\partial_{\bar{t}}\psi\cdot F|+|F|^2,
\end{equation}
on the right hand side of both equations of Theorem~\ref{thm: integrated local energy decay for psi}. Moreover, in this case, $\psi_\infty$ is bounded by 
\begin{equation}\label{eq: bound of psi infinity in the inhomogeneous case}
    |\psi_\infty|\leq C \left(\sup_{\phi_{\tau_1}(\Sigma)} |\psi| +\sqrt{\int_{\phi_{\tau_1(\Sigma)}} J^n_\mu[\psi] n^\mu+ \int\int_{D(0,\infty)}|\partial_{\bar{t}}\psi\cdot F|+|F|^2} \right).
\end{equation}
\end{remark}

\subsection{The equation for $\mathcal{G}\psi$ and the $\partial_{\bar{t}}$ energy identity}
Before we begin the proof, we derive the equation satisfied by $\mathcal{G}\psi$.

\begin{proposition}\label{prop: commuted equation for G psi}
Let $\psi$ satisfy the wave equation~\eqref{eq: wave equation}. Then, the following holds
\begin{equation}\label{eq: equation obeyed by Psi}
    \Box_{g_{M,\Lambda}}\mathcal{G}\psi=\big[\Box_{g_{M,\Lambda}},\mathcal{G}\big]\psi=\frac{2\sqrt{1-9M^2\Lambda}}{(1-\mu)r\sqrt{1+\frac{6M}{r}}}\partial_{\bar{t}}\mathcal{G}\psi+ E_1(r)\partial_{\bar{t}}\psi + E_2(r)\partial_r\psi 
\end{equation}
where 
\begin{equation}
    \begin{aligned}
        E_1(r) & =\frac{2\sqrt{1-\mu}(-1+2\Lambda r^2)}{r}, \\
        E_2(r) & = \frac{9M^2-24M\Lambda r^3+r^4\Lambda (9-2r^2\Lambda)}{9r^3\sqrt{1-\mu}}.
    \end{aligned}
\end{equation}
\end{proposition}
\begin{proof}
For convenience let 
\begin{equation}
    f(r)=r\sqrt{1-\mu}.
\end{equation}

Let $\varphi$ be a smooth function on $\mathcal{M}_{\textit{ext}}$, and note the commutation
\begin{equation}\label{eq: commutation for varphi}
    \begin{aligned}
        \big[\Box_{g_{M,\Lambda}},f(r)\frac{\partial}{\partial r}\big]\varphi &  =\Box(f(r)\partial_r\varphi)-f(r)\partial_r\Box_{g_{M,\Lambda}}\varphi=\Box_{g_{M,\Lambda}} f \partial_r\varphi+2\nabla^cf(r)\nabla_c \partial_r\varphi+f(r)\left(\Box_{g_{M,\Lambda}}\partial_r\varphi-\partial_r\Box_{g_{M,\Lambda}}\varphi\right) \\
        &   = \Box_{g_{M,\Lambda}} f(r)\partial_r\varphi +2g^{rr}\partial_r f(r)\partial_r^2\varphi+2g^{r\bar{t}}\partial_r f(r)\partial_{\bar{t}}\partial_r\varphi+ f(r)\left(\Box_{g_{M,\Lambda}}\partial_r\varphi-\partial_r\Box_{g_{M,\Lambda}}\varphi\right),
    \end{aligned}
\end{equation}
where 
\begin{equation}\label{eq: commutation partial r box varphi}
    \Box_{g_{M,\Lambda}}\partial_r\varphi-\partial_r\Box_{g_{M,\Lambda}}\varphi=-\left( \partial_r g^{rr} \partial_r^2\varphi +2\partial_r g^{r\bar{t}}\partial_{\bar{t}}\partial_r\varphi +\partial_rg^{\bar{t}\bar{t}}\partial_{\bar{t}}^2\varphi+\partial_r\left(\frac{1}{r^2}\right)r^2\slashed{\Delta}\varphi+\partial_r\left( g^{ab}\Gamma^c\:_{ab} \right)\partial_c\varphi \right).
\end{equation}

Now, by computing the right hand side of equation \eqref{eq: commutation partial r box varphi} for $\psi$, a solution of the wave equation \eqref{eq: wave equation}, we obtain
\begin{equation}\label{eq: commutation partial r box psi}
    \begin{aligned}
         & -\left( \partial_r g^{rr} \partial_r^2\psi +2\partial_r g^{r\bar{t}}\partial_{\bar{t}}\partial_r\psi +\partial_rg^{\bar{t}\bar{t}}\partial_{\bar{t}}^2\psi+\partial_r\left(\frac{1}{r^2}\right)r^2\slashed{\Delta}\psi+\partial_r\left( g^{ab}\Gamma^c\:_{ab} \right)\partial_c\psi \right) \\
        &  \quad =-\Bigg( \partial_r g^{rr} \partial_r^2\psi +2\partial_r g^{r\bar{t}}\partial_{\bar{t}}\partial_r\psi +\partial_rg^{\bar{t}\bar{t}}\partial_{\bar{t}}^2\psi\\
        &   \quad\quad\quad-\frac{2}{r}\left(-g^{rr}\partial_r^2\psi-2g^{r\bar{t}}\partial_{\bar{t}}\partial_r\psi-g^{\bar{t}\bar{t}}\partial_{\bar{t}}^2\psi +g^{ab}\Gamma^c\:_{ab}\partial_c\psi\right)+\partial_r\left(g^{ab}\Gamma^c\:_{ab}\right)\partial_c\psi \Bigg)\\
        & \quad =-\Bigg( \partial_r^2\psi\left(\partial_rg^{rr}+\frac{2}{r}g^{rr}\right) +\partial_{\bar{t}}^2\psi\left(\partial_rg^{\bar{t}\bar{t}}+\frac{2}{r}g^{\bar{t}\bar{t}}\right)+\partial_{\bar{t}}\partial_r\psi\left(2\partial_r g^{r\bar{t}}+\frac{4}{r}g^{r\bar{t}}\right) \\
        &   \quad\quad\quad +\frac{2}{r}g^{ab}\Gamma^c\:_{ab}\partial_c\psi +\partial_r\left(g^{ab}\Gamma^c\:_{ab}\right)\partial_c\psi\Bigg).
    \end{aligned}
\end{equation}
Finally, by using equation \eqref{eq: commutation for varphi} for $\psi$, a solution of the wave equation \eqref{eq: wave equation}, together with \eqref{eq: commutation partial r box psi}, we obtain
\begin{equation}
    \begin{aligned}
         \big[\Box_{g_{M,\Lambda}},f(r)\frac{\partial}{\partial r}\big]\psi &= \left(2\partial_r f(r) -f(r)\left(\partial_r g^{rr}+\frac{2}{r}g^{rr}\right)\right)\partial_r^2\psi+\left( 2g^{r\bar{t}}\partial_r f(r) -2 f(r)\left(\partial_rg^{r\bar{t}} +\frac{2}{r}g^{r\bar{t}}\right)\right)\partial_{{\bar{t}}}\partial_r\psi \\
         &  \quad -f(r)\left( \partial_rg^{\bar{t}\bar{t}}+\frac{2}{r}g^{\bar{t}\bar{t}} \right)\partial_{{\bar{t}}}^2\psi+\Box_{g_{M,\Lambda}} f(r) \partial_r\psi -\frac{2}{r}f(r)g^{ab}\Gamma^c\:_{ab}\partial_c\psi -f(r)\partial_r\left( g^{ab}\Gamma^c\:_{ab}\right)\partial_c\psi.
    \end{aligned}
\end{equation}
By revisiting the metric in regular hyperboloidal coordinates, see~\eqref{eq: regular metric de sitter}, we note 
\begin{equation}\label{eq: why the coordinates are useful, 1}
    \begin{aligned}
        \partial_rg^{\bar{t}\bar{t}}+\frac{2}{r}g^{\bar{t}\bar{t}}=0,
    \end{aligned}
\end{equation}
also
\begin{equation}\label{eq: why the coordinates are useful, 2}
    \begin{aligned}
        2\partial_r f(r)g^{rr} -f(r)\left(\partial_r g^{rr}+\frac{2}{r}g^{rr}\right)=0,
    \end{aligned}
\end{equation}
and finally 
\begin{equation}\label{eq: why the coordinates are useful, 3}
    2g^{r\bar{t}}\partial_r f(r)-2f(r)\left( \partial_rg^{r\bar{t}}+\frac{2}{r}g^{r\bar{t}}\right)= \frac{2\sqrt{1-9M^2\Lambda}}{\sqrt{1-\mu}\sqrt{1+\frac{6M}{r}}}.
\end{equation} 
Therefore, we conclude 
\begin{equation}\label{eq: in proof of thm 1, 1}
    \begin{aligned}
        \big[\Box_{g_{M,\Lambda}},f(r)\frac{\partial}{\partial r}\big]\psi &    = \frac{2\sqrt{1-9M^2\Lambda}}{\sqrt{1-\mu}\sqrt{1+\frac{6M}{r}}}\partial_r\partial_{{\bar{t}}}\psi + \Box_{g_{M,\Lambda}} f(r) \partial_r\psi -\frac{2}{r}f(r)g^{ab}\Gamma^c\:_{ab}\partial_c\psi -f(r)\partial_r\left( g^{ab}\Gamma^c\:_{ab}\right)\partial_c\psi \\
        &   = \frac{2\sqrt{1-9M^2\Lambda}}{r(1-\mu)\sqrt{1+\frac{6M}{r}}}\partial_{{\bar{t}}}\mathcal{G}\psi+ \Box_{g_{M,\Lambda}} f(r) \partial_r\psi -\frac{2}{r}f(r)g^{ab}\Gamma^c\:_{ab}\partial_c\psi -f(r)\partial_r\left( g^{ab}\Gamma^c\:_{ab}\right)\partial_c\psi.
    \end{aligned}
\end{equation}
Now, by revisiting the wave operator \eqref{eq: wave operator}, we compute 
\begin{equation}\label{eq: in proof of thm 1, 2}
    g^{ab}\Gamma^c\:_{ab}\partial_c\psi =\left(-\partial_r\xi +2r^{-1}\xi(r)\right)\partial_t\psi+\left(2r^{-1}(1-\mu) +\partial_r(1-\mu)\right)\partial_r\psi,
\end{equation}
and 
\begin{equation}\label{eq: in proof of thm 1, 3}
    \partial_r(g^{ab}\Gamma^c\:_{ab})\partial_c\psi=\partial_r\left(-\partial_r\xi+2r^{-1}\xi(r)\right)\partial_{\bar{t}}\psi+\partial_r\left(2r^{-1}(1-\mu)+\partial_r(1-\mu)\right)\partial_r\psi.
\end{equation}
By examining equations~\eqref{eq: in proof of thm 1, 1} and~\eqref{eq: in proof of thm 1, 2},~\eqref{eq: in proof of thm 1, 3} we conclude the result. 
\end{proof}

\begin{remark}\label{rem: remark on the choice of coordinates}
From relations \eqref{eq: why the coordinates are useful, 1}, \eqref{eq: why the coordinates are useful, 2} we can see how we arrived at our choice of coordinates, Definition \ref{def: regular hyperboloidal coordinates}, and our vector field \eqref{eq: new vector fields}. Specifically, in the class of coordinate tranformations $\bar{t}=t-H(r)$ we select $H(r)$ such that equation \eqref{eq: why the coordinates are useful, 1} holds. Moreover, we select the function $f$, such that \eqref{eq: why the coordinates are useful, 2} holds. By this choice of coordinates $(r,\bar{t},
\theta,\phi)$ and vector field $\mathcal{G}$, the expression of equation \eqref{eq: why the coordinates are useful, 3} is positive. This reflects the good unstable structure of trapping. 
\end{remark}

We apply a $\partial_{\bar{t}}$-multiplier estimate to the equation~\eqref{eq: equation obeyed by Psi}. 

\begin{proposition}\label{prop: estimate for Psi without all the derivatives}
Let $\psi$ satisfy the wave equation~\eqref{eq: wave equation}. Then, we have the following
\begin{equation}\label{eq: estimate for Psi without all the derivatives}
    \begin{aligned}
        &\left(\int_{\mathcal{H}^+\cap D(\tau_1,\tau_2)}+\int_{\bar{\mathcal{H}}^+\cap D(\tau_1,\tau_2)} +\int_{\phi_{\tau_2}(\Sigma)}\right)  J_\mu^{\partial_{\bar{t}}}[\mathcal{G}\psi]n^\mu +\int\int_{D(\tau_1,\tau_2)}\frac{2\sqrt{1-9M^2\Lambda}}{(1-\mu)r\sqrt{1+\frac{6M}{r}}}\left(\partial_{{\bar{t}}}\mathcal{G}\psi\right)^2\\
        & \quad =\int_{\phi_{\tau_1}(\Sigma)} J_\mu^{\partial_{\bar{t}}}[\mathcal{G}\psi]n^\mu -\int\int_{D(\tau_1,\tau_2)}\left(E_1(r) \partial_{r}\psi +E_2(r)\partial_{{\bar{t}}}\psi\right)\partial_{{\bar{t}}}\mathcal{G}\psi.
    \end{aligned}
\end{equation}
\end{proposition}
\begin{proof}
We have
\begin{equation}
    \nabla^\mu J_\mu^{\partial_{\bar{t}}}[\mathcal{G}\psi]=\partial_{{\bar{t}}} \mathcal{G}\psi\Box_{g_{M,\Lambda}}\mathcal{G}\psi,
\end{equation}
and we have already computed
\begin{equation}
    \Box_{g_{M,\Lambda}}\mathcal{G}\psi=\frac{2\sqrt{1-9M^2\Lambda}}{(1-\mu)r\sqrt{1+\frac{6M}{r}}}\partial_{{\bar{t}}}\mathcal{G}\psi+ E_1(r)\partial_{r}\psi+ E_2(r)\partial_{{\bar{t}}}\psi
\end{equation}
in Proposition~\ref{prop: commuted equation for G psi}. We apply the divergence theorem, see equation~\eqref{eq: divergence Theorem}, to $\mathcal{G}\psi$ and conclude the result. 
\end{proof}

\subsection{Auxilliary estimates}

\begin{remark}\label{rem: keeping track of the powers of r}
In the proof of this Section we shall keep track of the powers of $r$ in our estimates, such that the constants in our estimates do not degenerate as $\Lambda\rightarrow 0$. This shall be useful in the next Section \ref{sec: thm in Schwarzschild}, where we treat the $\Lambda=0$ Schwarzschild case.
\end{remark}

We will also need the following elementary pointwise lemma later. 
\begin{lemma}\label{lem: the small order terms in the main estimate}
Let $E_1(r),E_2(r)$ be defined as in Proposition~\ref{prop: commuted equation for G psi}. Then, we have the following
\begin{equation}
    \begin{aligned}
        \Big|\int\int_{D(\tau_1,\tau_2)}& \left(E_1(r)\partial_{{\bar{t}}}\psi+E_2(r)\partial_{r}\psi\right)\partial_{{\bar{t}}}\mathcal{G}\psi\Big|\\
        &   \leq \int\int_{D(\tau_1,\tau_2)}\left(\frac{r}{2\epsilon}(1-\mu)\left(E_2(r)\partial_{r}\psi\right)^2+\frac{\epsilon}{2}\frac{1}{(1-\mu)r}(\partial_{{\bar{t}}}\mathcal{G}\psi)^2\right)\\
        &   \quad+\int\int_{D(\tau_1,\tau_2)}|\partial_r(1-\mu)|(\partial_{\bar{t}}\psi)^2.
    \end{aligned}
\end{equation}
\end{lemma}

We want to generate all the derivatives of $\mathcal{G}\psi$ from terms appearing in equation~\eqref{eq: estimate for Psi without all the derivatives}.

\begin{lemma}\label{lem: control of Psi prime}
There exists a constant $C(M,\Lambda)>0$, which does not degenerate as $\Lambda\rightarrow 0$, such that, if $\psi$ solves the wave equation~\eqref{eq: wave equation}, then we obtain
\begin{equation}
    \begin{aligned}
        \int\int_{D(\tau_1,\tau_2)}&  \frac{1-\mu}{r}\left(\partial_r\mathcal{G}\psi\right)^2+\frac{1}{r}|\slashed{\nabla}\mathcal{G}\psi|^2  \\
        &\leq C\Bigg( \int\int_{D(\tau_1,\tau_2)} r^{-1}(1-\mu)^{-1}\left(\partial_{\bar{t}}\mathcal{G}\psi\right)^2 +r^{-2}(\partial_r \psi)^2 + r^{-2} (\partial_t\psi)^2 \\
        & \quad\quad\quad  +\int_{\phi_{\tau_2}(\Sigma)} r^{-1}(\partial_{{\bar{t}}}\mathcal{G}\psi)^2+ (1-\mu)r(\partial_r\mathcal{G}\psi)^2+r(\partial_r\psi)^2 \\
        &   \quad\quad\quad+\int_{\phi_{\tau_1}(\Sigma)}r^{-1}(\partial_{{\bar{t}}}\mathcal{G}\psi)^2+ (1-\mu)r(\partial_r\mathcal{G}\psi)^2+r(\partial_r\psi)^2 \Bigg).
    \end{aligned}
\end{equation}
\end{lemma}
\begin{proof}
We begin from the equation for $\mathcal{G}\psi$, see the commutation~\eqref{eq: equation obeyed by Psi}
\begin{equation}\label{eq: Gpsi in Lemma that controls derivatives}
    \begin{aligned}
    &   -\frac{1-\xi^2(r)}{1-\mu}\partial_{\bar{t}}^2\mathcal{G}\psi+(1-\mu)\partial_r^2\mathcal{G}\psi+2\xi(r)\partial_r\partial_{\bar{t}}\mathcal{G}\psi+ \slashed{\nabla}^A\slashed{\nabla}_A \mathcal{G}\psi -\frac{2\sqrt{1-9M^2\Lambda}}{(1-\mu)r\sqrt{1+\frac{6M}{r}}}\partial_{\bar{t}}\mathcal{G}\psi \\
    \quad &= \left(E_1(r)-\left(\partial_r\xi -2r^{-1}\xi(r)\right)\right)\partial_{\bar{t}}\psi  +\left(E_2(r)-\left(2r^{-1}(1-\mu)+\partial_r(1-\mu)\right)\right)\partial_{r}\psi
    \end{aligned}
\end{equation}
and multiply it with $\frac{1}{r}\mathcal{G}\psi$
\begin{equation}\label{eq: Gpsi in Lemma that controls derivatives,1}
    \begin{aligned}
    -\frac{1-\xi^2(r)}{1-\mu}\frac{1}{r}\mathcal{G}\psi\partial_t^2\mathcal{G}\psi & +(1-\mu)\frac{1}{r}\mathcal{G}\psi\partial_r^2\mathcal{G}\psi+2\xi(r)\frac{1}{r}\mathcal{G}\psi\partial_r\partial_{\bar{t}}\mathcal{G}\psi+ \frac{1}{r}\mathcal{G}\psi\slashed{\nabla}^A\slashed{\nabla}_A \mathcal{G}\psi \\
    &   -\frac{2\sqrt{1-9M^2\Lambda}}{(1-\mu)r\sqrt{1+\frac{6M}{r}}}\frac{1}{r}\mathcal{G}\psi\partial_{\bar{t}}\mathcal{G}\psi = \frac{1}{r}\mathcal{G}\psi \textit{RHS}(\psi)
    \end{aligned}
\end{equation}
where $\textit{RHS}$ is the right hand side of equation~\eqref{eq: Gpsi in Lemma that controls derivatives}.

We integrate~\eqref{eq: Gpsi in Lemma that controls derivatives,1} over $D(\tau_1,\tau_2)$. Then, we perform the necessary integration by parts and Young's inequalities to get the result. 
\end{proof}

We want to generate a non-degenerate $(\partial_t\psi)^2$ from terms appearing in equation~\eqref{eq: estimate for Psi without all the derivatives}.

\begin{lemma}\label{lem: the Psi dot times psi dot}
There exists a constant $C(M,\Lambda)>0$, that does not degenerate as $\Lambda\rightarrow 0$, such that for all $\epsilon>0$, if $\psi$ satisfies the wave equation~\eqref{eq: wave equation}, then we have
\begin{equation*}
    \begin{aligned}
        \frac{C}{\epsilon}\int\int_{D(\tau_1,\tau_2)} \frac{1-\mu}{r^3}(\partial_{{\bar{t}}}\psi)^2\leq \frac{1}{2}&\int\int_{D(\tau_1,\tau_2)}\frac{2\sqrt{1-9M^2\Lambda}}{r^4(1-\mu)\sqrt{1+\frac{6M}{r}}}(\partial_{{\bar{t}}}\mathcal{G}\psi)^2 \\
        &   +\frac{C}{\epsilon^2}\left(\int\int_{D(\tau_1,\tau_2)}\frac{(r-3M)^2}{r^4\sqrt{1-9M^2\Lambda}}(\partial_{{\bar{t}}}\psi)^2\right). 
    \end{aligned}
\end{equation*}
\end{lemma}
\begin{proof}
We begin by 
\begin{equation*}
        \int\int_{D(\tau_1,\tau_2)}\frac{1-\mu}{r^2}(\partial_{\bar{t}}\psi)^2\frac{1}{r}dg= \int_{\mathbb{S}^2}\int_0^\tau \int_{r_+}^{\bar{r}_+}\frac{1-\mu}{r^2}\frac{\partial}{\partial r} (r-3M) (\partial_t\psi)^2 \sin{\theta}dr d\bar{t}d\sigma.
\end{equation*}        
We perform an integration by parts, and write the right hand side of the above as 
\begin{equation*}
    \begin{aligned}
        & \int_{\mathbb{S}^2}\int_0^\tau\big[\frac{1-\mu}{r^2}(r-3M)(\partial_{\bar{t}}\psi)^2\big]_{r_+}^{\bar{r}_+}\\
        &   \quad -\int\int_{D(\tau_1,\tau_2)}\left(\frac{\partial}{\partial r} \left(\frac{1-\mu}{r^2}\right)(r-3M)(\partial_{\bar{t}}\psi)^2+(r-3M)\frac{1-\mu}{r^2}2\partial_{\bar{t}}\psi\frac{\partial}{\partial r} \partial_{\bar{t}}\psi\right)r\sin{\theta}dr d\bar{t}d\sigma\\
        &   =-\int\int_{D(\tau_1,\tau_2)}\left(\frac{\partial}{\partial r} \left(\frac{1-\mu}{r^2}\right)(r-3M)(\partial_{\bar{t}}\psi)^2+(r-3M)\frac{1-\mu}{r^2}2\partial_{\bar{t}}\psi\frac{\partial}{\partial r} \partial_{\bar{t}}\psi\right)r\sin{\theta}dr d\bar{t}d\sigma \\
        &   \quad\leq C\int\int_{D(\tau_1,\tau_2)}  \frac{1}{r^6}(r-3M)^2 (\partial_{\bar{t}}\psi)^2+\frac{1}{\epsilon r^8\sqrt{1-9M^2\Lambda}}(r-3M)^2(\partial_{\bar{t}}\psi)^2+\epsilon \frac{\sqrt{1-9M^2\Lambda}}{r^4(1-\mu)\sqrt{1+\frac{6M}{r}}}(\partial_{\bar{t}}\mathcal{G}\psi)^2.
    \end{aligned}
\end{equation*}

We conclude the proof.
\end{proof}

\subsection{Proofs of Theorem~\ref{main theorem 1} and Corollary~\ref{corollary on exponential decay}}

We are now ready to conclude Theorem~\ref{main theorem 1} and the exponential decay Corollary~\ref{corollary on exponential decay}. 
\begin{proof}[\textbf{Proof of Theorem~\ref{main theorem 1}}]
We first use Lemma~\ref{lem: the small order terms in the main estimate} on equation~\eqref{eq: estimate for Psi without all the derivatives}, in conjunction with Lemmata~\ref{lem: the small order terms in the main estimate}, \ref{lem: the Psi dot times psi dot} to obtain
\begin{equation}\label{eq: last inequality before the proof of theorem 1, 0}
    \begin{aligned}
        \Big(\int_{\mathcal{H}^+\cap D(\tau_1,\tau_2)}&+\int_{\bar{\mathcal{H}}^+\cap D(\tau_1,\tau_2)} +\int_{\phi_{\tau_2}(\Sigma)}\Big) J_\mu^{\partial_{\bar{t}}}[\mathcal{G}\psi]n^\mu \\
        &   +\int\int_{D(\tau_1,\tau_2)}\frac{2\sqrt{1-9M^2\Lambda}}{(1-\mu)r\sqrt{1+\frac{6M}{r}}}\left(\partial_{{\bar{t}}}\mathcal{G}\psi\right)^2+\frac{1-\mu}{r^4}(\partial_{\bar{t}}\psi)^2\\
        & \quad \lesssim \int_{\phi_{\tau_1}(\Sigma)} J_\mu^{\partial_{\bar{t}}}[\mathcal{G}\psi]n^\mu +\int\int_{D(\tau_1,\tau_2)}r^{-2}(\partial_{\bar{t}}\psi)^2+r^{-2}(\partial_r\psi)^2
    \end{aligned}
\end{equation}
Note that the horizon hypersurface terms on $\mathcal{H}^+,\bar{\mathcal{H}}^+$, on the left hand side of equation \eqref{eq: last inequality before the proof of theorem 1, 0}, are non-negative, see \eqref{eq: nonnegative energy momentum tensor}, so we may drop them from our estimates. We multiply the estimate of Lemma~\ref{lem: control of Psi prime} with a smallness parameter and add it to \eqref{eq: last inequality before the proof of theorem 1, 0} to obtain
\begin{equation}\label{eq: last inequality before the proof of theorem 1}
    \begin{aligned}
        &   \int_{\phi_{\tau_2}(\Sigma)} \left(J_\mu^{\partial_{\bar{t}}}[\mathcal{G}\psi]n^\mu\right) \\& +\int\int_{D(\tau_1,\tau_2)}\frac{1}{r(1-\mu)}(\partial_{{\bar{t}}}\mathcal{G}\psi)^2+ \frac{1-\mu}{r}(\partial_r\mathcal{G}\psi)^2+\frac{1}{r}|\slashed{\nabla}\mathcal{G}\psi|^2 + \frac{1-\mu}{r^4}(\partial_{\bar{t}}\psi)^2\\
        &   \quad \lesssim \int_{\phi_{\tau_1}(\Sigma)} J_\mu^{\partial_{\bar{t}}}[\mathcal{G}\psi]n^\mu +\int\int_{D(\tau_1,\tau_2)}r^{-2}(\partial_{\bar{t}}\psi)^2+r^{-2}(\partial_r\psi)^2 \\
        &  \quad\quad  +\int_{\phi_{\tau_2}(\Sigma)}(\partial_r\psi)^2+\int_{\phi_{\tau_1}(\Sigma)} (\partial_r\psi)^2.
    \end{aligned}
\end{equation}
Finally, we add in the Morawetz estimate~\eqref{eq: integrated local energy decay for psi} of Theorem~\ref{thm: integrated local energy decay for psi}, multiplied with a large parameter and use the boundedness estimate~\eqref{eq: boundedness estimate}, to conclude
\begin{equation}
    \begin{aligned}
        &   \int_{\phi_{\tau_2}(\Sigma)}\left(J_\mu^{\partial_{\bar{t}}}[\mathcal{G}\psi]n^\mu+J^n_\mu[\psi]n^\mu\right) \\& +\int\int_{D(\tau_1,\tau_2)}\frac{1}{(1-\mu)}(\partial_{{\bar{t}}}\mathcal{G}\psi)^2+ (1-\mu)(\partial_r\mathcal{G}\psi)^2+|\slashed{\nabla}\mathcal{G}\psi|^2+ (\partial_r\psi)^2 +(\partial_{{\bar{t}}}\psi)^2+|\slashed{\nabla}\psi|^2 \\
        &   \quad \leq C\int_{\phi_{\tau_1}(\Sigma)} J_\mu^{\partial_{\bar{t}}}[\mathcal{G}\psi]n^\mu+J^n_\mu[\psi]n^\mu.
    \end{aligned}
\end{equation}
We recall the definition of $\mathcal{E}(\mathcal{G}\psi,\psi)$ in equation \eqref{eq: new energy} and note equation \eqref{eq: new energy similarity}. We conclude~\eqref{eq: estimate of thm 1, 1}. 

Moreover, we conclude~\eqref{eq: estimate of thm 1, 2}, in view of the relation of the coaread formula of Section~\ref{subsec: coarea formula}, namely 
\begin{equation}
    \int\int_{D(\tau_1,\tau_2)}\mathcal{E}(\mathcal{G}\psi,\psi)\sim \int_{\tau_1}^{\tau_2}d\tau\int_{\phi_\tau(\Sigma)}\mathcal{E}(\mathcal{G}\psi,\psi).
\end{equation}
\end{proof}
We need the following lemma. 

\begin{lemma}\label{lem: similar sides give exp decay}
Let $f:[0,\infty)\rightarrow\mathbb{R}$ a non-negative continuous function satisfying 
\begin{equation*}
    f(\tau_2)+\int_{\tau_1}^{\tau_2}f(\tau) d\tau \leq k f(\tau_1),
\end{equation*}
for all $\tau_2>\tau_1\geq 0$ and some $k> 0$. Then, there exists constants $c,C>0$ depending only on $k$ such that: 
\begin{equation*}
    f(\tau)\leq C  e^{- c\tau}f(0),
\end{equation*}
for all $\tau\geq 0$.
\end{lemma}
\begin{proof}
This is elementary. See for example~\cite{DR3}. 
\end{proof}

Now we can infer exponential decay.
\begin{proof}[\textbf{Proof of Corollary~\ref{corollary on exponential decay}}]
The estimate 
\begin{equation}\label{eq: estimate of cor 1, 1. In the proof}
    \int_{\phi_{\tau}(\Sigma)}\mathcal{E}(\mathcal{G}\psi,\psi)\lesssim e^{-c\tau}\int_{\Sigma}\mathcal{E}(\mathcal{G}\psi,\psi)
\end{equation}
is an immediate consequence of equation~\eqref{eq: estimate of thm 1, 2} of Theorem~\ref{main theorem 1} and Lemma~\ref{lem: similar sides give exp decay}, where $f(\tau)=\int_{\phi_\tau(\Sigma)}\mathcal{E}(\mathcal{G}\psi,\psi)$. 

We note that inequality \eqref{eq: estimate of cor 1, 1. In the proof} holds for $\partial_{\bar{t}}^i\Omega_\alpha^j\psi$, for all indices $i,j$, in the place of $\psi$, where $\Omega_\alpha$ are defined in equation \eqref{eq: generators of lie algebra}, since the following hold: $[\Box_{g_{M,\Lambda}},\partial_{\bar{t}}]=0, [\Box_{g_{M,\Lambda}},\Omega_\alpha]=0$. 

Therefore, one can prove, by commuting with $\Omega_\alpha$ and then a Sobolev estimate, the pointwise estimate  
\begin{equation}
    \sup_{\phi_{\tau}(\Sigma)} |\psi-\psi_\infty|^2 \lesssim \int_{\phi_\tau(\Sigma)} \sum_{0\leq i\leq 1}\sum_{0\leq \alpha\leq 3} \mathcal{E}(\mathcal{G}\Omega_\alpha^i\psi,\Omega_\alpha\psi)  \lesssim  e^{-c{\tau}} \int_\Sigma \sum_{0\leq i\leq 1}\sum_{0\leq \alpha\leq 3} \mathcal{E}(\mathcal{G}\Omega_\alpha^i\psi,\Omega_\alpha\psi),
\end{equation}
where $\psi_\infty$ is from Theorem \ref{thm: integrated local energy decay for psi}, and satisfies $|\psi_\infty|\leq C\left(\sup_{\Sigma}|\psi|+\sqrt{\int_{\Sigma}J^n_\mu[\psi]n^\mu}  \right)$.
\end{proof}

\begin{remark}
Note that from Theorem~\ref{thm: higher order G estimate} and the use of Sobolev estimates we can obtain pointwise estimates for arbitrary regular higher order derivatives
\begin{equation}
    \sup_{\phi_\tau(\Sigma)}\sum_{1\leq i+j+l\leq k} |\partial_r^i\partial_{\bar{t}}^j\Omega_\alpha^l\psi|^2\lesssim E_{\mathcal{G},k+3} e^{-c\tau},
\end{equation}
where for $E_{\mathcal{G},k+3}$ see \eqref{eq: cor: higher order G estimate, eq 2}.
\end{remark}

\section{Proof of Theorem~\ref{thm: absorption of small terms}}\label{sec: proof of main theorem 2}

We recall
\begin{equation}
    \big[\Box_{g_{M,\Lambda}},\mathcal{G}\big]\psi =\frac{2\sqrt{1-9M^2\Lambda}}{(1-\mu)r\sqrt{1+\frac{6M}{r}}}\partial_{{\bar{t}}}\mathcal{G}\psi + E_1(r)\partial_{\bar{t}}\psi + E_2(r)\partial_r\psi 
\end{equation}
from which we easily deduce 
\begin{equation}\label{eq: Gpsi,  inhomogeneous case}
    \Box_{g_{M,\Lambda}}\mathcal{G}\psi= \frac{2\sqrt{1-9M^2\Lambda}}{(1-\mu)r\sqrt{1+\frac{6M}{r}}}\partial_{{\bar{t}}}\mathcal{G}\psi+E_1(r)\partial_{\bar{t}}\psi + E_2(r)\partial_r\psi +\mathcal{G} F.
\end{equation}
Now, we follow the arguments of the Section~\ref{sec: proof of main theorem 1}. Specifically, we conclude that there exists a constant $C=C(M,\Lambda)>0$, such that for $\psi$ a sufficiently regular solution of the inhomogeneous wave equation~\eqref{eq: wave equation with error terms} on $D(\tau_1,\tau_2)$, we obtain
\begin{equation}\label{eq: Gpsi,  inhomogeneous case, eq 1}
	\begin{aligned}
		&   \int_{\phi_{\tau_2}(\Sigma)}  J_\mu^{\partial_{\bar{t}}}[\mathcal{G}\psi]n^\mu \\& +\int\int_{D(\tau_1,\tau_2)}\frac{1}{r(1-\mu)}(\partial_{{\bar{t}}}\mathcal{G}\psi)^2\\
		&   \quad \leq C\int_{\phi_{\tau_1}(\Sigma)} J_\mu^{\partial_{\bar{t}}}[\mathcal{G}\psi]n^\mu + C\int\int_{D(\tau_1,\tau_2)} \left| \partial_{\bar{t}}\mathcal{G}\psi\mathcal{G}F \right|+C\int\int_{D(\tau_1,\tau_2)} \left|\partial_{\bar{t}}\mathcal{G}\psi\left(E_1(r)\partial_{\bar{t}}\psi+E(r)\partial_r\psi\right)\right|.
	\end{aligned}
\end{equation}
Then, we use appropriate Young's inequalities on the right hand side of~\eqref{eq: Gpsi,  inhomogeneous case, eq 1} to conclude 
\begin{equation}\label{eq: Gpsi,  inhomogeneous case, eq 2}
\begin{aligned}
&   \int_{\phi_{\tau_2}(\Sigma)}  J_\mu^{\partial_{\bar{t}}}[\mathcal{G}\psi]n^\mu \\& +\int\int_{D(\tau_1,\tau_2)}\frac{1}{r(1-\mu)}(\partial_{{\bar{t}}}\mathcal{G}\psi)^2 \\
&   \quad \leq C\int_{\phi_{\tau_1}(\Sigma)} J_\mu^{\partial_{\bar{t}}}[\mathcal{G}\psi]n^\mu + C\int\int_{D(\tau_1,\tau_2)} (1-\mu)\left|\mathcal{G}F \right|^2+C\int\int_{D(\tau_1,\tau_2)} (1-\mu)\left|E_1(r)\partial_{\bar{t}}\psi+E_2(r)\partial_r\psi\right|^2.
\end{aligned}
\end{equation}
Now, we sum in~\eqref{eq: Gpsi,  inhomogeneous case, eq 2} the Morawetz estimate of Theorem~\ref{thm: integrated local energy decay for psi}, (in the inhomogeneous form given by Remark~\ref{rem: on the inhomogeneous Morawetz}), to conclude
\begin{equation}\label{eq: the last inequality, in S}
\begin{aligned}
&   \int_{\phi_{\tau_2}(\Sigma)} \left(J_\mu^{\partial_{\bar{t}}}[\mathcal{G}\psi]n^\mu+J^n_\mu[\psi]n^\mu\right) \\& +\int\int_{D(\tau_1,\tau_2)}\frac{1}{r(1-\mu)}(\partial_{{\bar{t}}}\mathcal{G}\psi)^2+ \frac{1-\mu}{r}(\partial_r\mathcal{G}\psi)^2+\frac{1}{r}|\slashed{\nabla}\mathcal{G}\psi|^2+ (\partial_r\psi)^2 +(\partial_{{\bar{t}}}\psi)^2+|\slashed{\nabla}\psi|^2 \\
&   \quad \leq C\int_{\phi_{\tau_1}(\Sigma)} J_\mu^{\partial_{\bar{t}}}[\mathcal{G}\psi]n^\mu+J^n_\mu[\psi]n^\mu + C\int\int_{D(\tau_1,\tau_2)} (1-\mu)\left|\mathcal{G}F \right|^2+|\partial_{\bar{t}}\psi\cdot F|,
\end{aligned}
 \end{equation}
where we generated all the derivatives of $\mathcal{G}\psi$ and $\psi$ on the left hand side of our estimate~\eqref{eq: the last inequality, in S}, by Lemmata~\ref{lem: control of Psi prime},~\ref{lem: the Psi dot times psi dot}. We apply a Young's inequality on the last term on the right hand side of~\eqref{eq: the last inequality, in S} and conclude 
\begin{equation}\label{eq: the last inequality, in S, eq 1}
\begin{aligned}
&   \int_{\phi_{\tau_2}(\Sigma)} \left(J_\mu^{\partial_{\bar{t}}}[\mathcal{G}\psi]n^\mu+J^n_\mu[\psi]n^\mu\right) \\& +\int\int_{D(\tau_1,\tau_2)}\frac{1}{r(1-\mu)}(\partial_{{\bar{t}}}\mathcal{G}\psi)^2+ \frac{1-\mu}{r}(\partial_r\mathcal{G}\psi)^2+\frac{1}{r}|\slashed{\nabla}\mathcal{G}\psi|^2+ (\partial_r\psi)^2 +(\partial_{{\bar{t}}}\psi)^2+|\slashed{\nabla}\psi|^2 \\
&   \quad \leq C\int_{\phi_{\tau_1}(\Sigma)} J_\mu^{\partial_{\bar{t}}}[\mathcal{G}\psi]n^\mu+J^n_\mu[\psi]n^\mu + C\int\int_{D(\tau_1,\tau_2)} (1-\mu)\left|\mathcal{G}F \right|^2+|F|^2.
\end{aligned}
\end{equation}

Now, Theorem~\ref{thm: absorption of small terms} is a trivial consequence of equation~\eqref{eq: the last inequality, in S, eq 1}, since we have the property \eqref{eq: new energy similarity}.

\begin{proof}[\textbf{Proof of Corollary~\ref{cor: absorption of small terms}}]
Suppose that
\begin{equation*}
    F=\epsilon a^j\partial_j\psi,
\end{equation*}
where $a^j$ are smooth and bounded with $\mathcal{G}a^j$ bounded, for all $j$. Now, we use equation~\eqref{eq: the last inequality, in S} to obtain 
\begin{equation}\label{eq: last inequality in Cor 2}
    \begin{aligned}
        &   \int_{\phi_{\tau_2}(\Sigma)} \left(J_\mu^{\partial_{\bar{t}}}[\mathcal{G}\psi]n^\mu+J^n_\mu[\psi]n^\mu\right) \\& +\int\int_{D(\tau_1,\tau_2)}\frac{1}{r(1-\mu)}(\partial_{{\bar{t}}}\mathcal{G}\psi)^2+ \frac{1-\mu}{r}(\partial_r\mathcal{G}\psi)^2+\frac{1}{r}|\slashed{\nabla}\mathcal{G}\psi|^2+ (\partial_r\psi)^2 +(\partial_{{\bar{t}}}\psi)^2+|\slashed{\nabla}\psi|^2 \\
        &   \quad \leq C\int_{\phi_{\tau_1}(\Sigma)} J_\mu^{\partial_{\bar{t}}}[\mathcal{G}\psi]n^\mu+J^n_\mu[\psi]n^\mu \\
        &   \quad\quad + C\int\int_{D(\tau_1,\tau_2)} \epsilon^2\left(\partial_j\psi\mathcal{G}a^j \right)^2+\epsilon^2(\mathcal{G}\psi)^2+\epsilon ^2(a^j\mathcal{G}\partial_j\psi)^2
    \end{aligned}
\end{equation}
since
\begin{equation*}
    \begin{aligned}
        \mathcal{G} F &= \epsilon\left(\partial_j\psi \mathcal{G} a^j\right)+ \epsilon\left(a^j\mathcal{G}\partial_j\psi\right).
    \end{aligned}
\end{equation*}
If $\epsilon$ is sufficiently small the terms on the right hand side can be absorbed, after also using a Hardy inequality. 

Finally, for the pointwise result, we commute the inhomogeneous equation \eqref{eq: wave equation with error terms} with the vector fields $\Omega_\alpha$, see equation \eqref{eq: generators of lie algebra}, to conclude that \eqref{eq: last inequality in Cor 2} holds for $\Omega_\alpha\psi$ in the place of $\psi$. Moreover, we know that $[\Omega_\alpha\mathcal{G}, a]^{\bar{t}},[\Omega_\alpha\mathcal{G}, a]^{r},g^{\theta\theta}([\Omega_\alpha\mathcal{G}, a]^{\theta})^2+g^{\phi\phi}([\Omega_\alpha\mathcal{G}, a]^{\phi})^2$ are bounded. We use a Sobolev inequality and conclude that 
\begin{equation}
    \sup_{\phi_\tau(\Sigma)}|\psi-\psi_\infty|\leq C\sqrt{E}e^{-c\tau}
\end{equation}
where, in view of equation \eqref{eq: bound of psi infinity in the inhomogeneous case}, the following holds
\begin{equation}
|\psi_\infty|\leq C\left(\sup_\Sigma |\psi|+\sqrt{\int_\Sigma \mathbb{T}(n,n)[\psi]+\mathbb{T}(\partial_{\bar{t}},n)[\mathcal{G}\psi]}\right)
\end{equation}
and 
\begin{equation}
    E=\sum_{0\leq i\leq 1}\sum_{1\leq \alpha\leq 3}\int_{\Sigma}\mathcal{E}(\mathcal{G}\Omega_\alpha^i\psi,\Omega_\alpha^i\psi).
\end{equation}
\end{proof}

\section{Proof of Theorem \ref{thm: higher order G estimate}}\label{sec: higher order G estimate}

	We define the auxiliary energy 
	\begin{equation}
	E_{j}[\psi](\tau)=  \sum_{1\leq i_1+i_2+i_3\leq j}\sum_{\alpha=1,2,3}\int_{\phi_\tau(\Sigma)}\left(\partial_{\bar{t}}^{i_1}\partial_r^{i_2}\Omega_\alpha^{i_3}\psi\right)^2.
	\end{equation}	
	
	We begin by noting the following Proposition, which is a higher order analogue of the inhomogeneous version of Theorem~\ref{thm: integrated local energy decay for psi}

	\begin{proposition}\label{prop: sec: higher order G estimate, prop 1}
	Let $\psi$ satisfy the inhomogeneous wave equation~\eqref{eq: cor: higher order G estimate, eq 0}. Then, we obtain the following higher order Morawetz estimate 
	\begin{equation}\label{eq: sec: higher order G, proof eq 1.1}
	\begin{aligned}
	&	E_j[\psi](\tau_2) + \int\int_{D(\tau_1,\tau_2)} \left(1-\frac{3M}{r}\right)^2\sum_{i_1+i_2=j}\sum_{\alpha} \left(\partial_{\bar{t}}^{i_1}\Omega_\alpha^{i_2}\psi\right)^2+\sum_{i_1+i_2+i_3=j,i_3\geq 1}\sum_\alpha\left(\partial_{\bar{t}}^{i_1}\Omega_\alpha^{i_2}\partial_r^{i_3}\psi\right)^2\\
	&	\qquad\quad\quad\quad\quad\quad\quad+\sum_{1\leq i_1+i_2+i_3\leq j-1}\sum_\alpha\left(\partial_{\bar{t}}^{i_1}\Omega_\alpha^{i_2}\partial_r^{i_3}\psi\right)^2\\
	&	\leq C E_j[\psi](\tau_1) +C\int\int_{D(\tau_1,\tau_2)}\left|\partial_{\bar{t}}^{j}\psi\cdot \partial_{\bar{t}}^{j-1}F\right|+\sum_{1\leq i_1+i_2+i_3\leq j-2}\sum_\alpha\left(\partial_{\bar{t}}^{i_1}\Omega_\alpha^{i_2}\partial_r^{i_3}F\right)^2\\
	&	\qquad\qquad\qquad +C\int_{\phi_{\tau_2}(\Sigma)} \sum_{0\leq i_1+i_2+i_3\leq j-2}\sum_\alpha\left(\partial_{\bar{t}}^{i_1}\Omega_\alpha^{i_2}\partial_r^{i_3}F\right)^2
	\end{aligned}
	\end{equation}
	for a constant $C(j,M,\Lambda)>0$.
	\end{proposition}
	\begin{proof}
	The proof of this Proposition follows from Theorem~\ref{thm: integrated local energy decay for psi} and additional redshift commutations of the Lecture notes~\cite{DR5} and elliptic estimates.
	\end{proof}

	Now we prove Theorem~\ref{thm: higher order G estimate}. 
	
	\begin{proof}[\textbf{Proof of Theorem \ref{thm: higher order G estimate}}]
	
	We start by recalling the result of Theorem \ref{thm: absorption of small terms} namely
	\begin{equation}\label{eq: sec: higher order G, proof eq 1}
		\begin{aligned}
			&	\int_{\phi_{\tau_2}(\Sigma)}\mathcal{E}\left(\mathcal{G}\psi,\psi\right)+\int\int_{D(\tau_1,\tau_2)}\frac{1}{r(1-\mu)}(\partial_t\mathcal{G}\psi)^2+\frac{1-\mu}{r}(\partial_r\mathcal{G}\psi)^2+\frac{1}{r}\left|\slashed{\nabla}\mathcal{G}\psi\right|^2+(\partial_{\bar{t}}\psi)^2+(\partial_r\psi)^2+|\slashed{\nabla}\psi|^2\\
			&	\quad\quad\quad\quad\quad\leq C \int_{\phi_{\tau_1}(\Sigma)} \mathcal{E}\left(\mathcal{G}\psi,\psi\right)+C\int\int_{D(\tau_1,\tau_2)}(1-\mu)r\left|\mathcal{G}F\right|^2+|F|^2,
		\end{aligned}
	\end{equation}
	 where we have kept certain $r$ factor explicitly for integrands related to $\mathcal{G}$, for comparison with the case $\Lambda=0$. (We note however that we drop the $r$ factor in the lower order terms.) The inequality~\eqref{eq: sec: higher order G, proof eq 1} already gives the result of the Theorem for $j=2$.

	 Now, for any $j\geq 3$, by commuting the inhomogeneous wave equation~\eqref{eq: cor: higher order G estimate, eq 0} with 
	\begin{equation}
		\partial_{\bar{t}}^k\Omega_\alpha^l,\qquad k+l\leq j-2
	\end{equation}
	we obtain 
	\begin{equation}\label{eq: sec: higher order G, proof eq 1.09}
	\begin{aligned}
	&	\int_{\phi_{\tau_2}(\Sigma)}\sum_{0\leq k+l= j-2}\sum_\alpha\mathcal{E}\left(\partial_{\bar{t}}^k\Omega_\alpha^l\mathcal{G}\psi,\partial_{\bar{t}}^k\Omega_\alpha^l\psi\right)\\
	&	+\int\int_{D(\tau_1,\tau_2)}\sum_{0\leq k+l= j-2}\sum_\alpha\left(\frac{1}{r(1-\mu)}(\partial_t\partial_{\bar{t}}^{k}\Omega_\alpha^{l}\mathcal{G}\psi)^2+\frac{1-\mu}{r}(\partial_r\partial_{\bar{t}}^{k}\Omega_\alpha^{l}\mathcal{G}\psi)^2+\frac{1}{r}\left|\slashed{\nabla}\partial_{\bar{t}}^{k}\Omega_\alpha^{l}\mathcal{G}\psi\right|^2\right)\\
	&	\quad\quad\quad\quad\quad\quad+\sum_{0\leq k+l= j-2}\sum_\alpha\left((\partial_{\bar{t}}\partial_{\bar{t}}^k\Omega_\alpha^l\psi)^2+(\partial_r\partial_{\bar{t}}^k\Omega_\alpha^l\psi)^2+|\slashed{\nabla}\partial_{\bar{t}}^k\Omega_\alpha^l\psi|^2\right)\\
	&	\leq C \int_{\phi_{\tau_1}(\Sigma)}\sum_{0\leq k+l= j-2}\sum_\alpha\mathcal{E}\left(\partial_{\bar{t}}^k\Omega_\alpha^l\mathcal{G}\psi,\partial_{\bar{t}}^k\Omega_\alpha^l\psi\right)\\
	&	\qquad+C\int\int_{D(\tau_1,\tau_2)}\sum_{0\leq k+l= j-2}\sum_\alpha\left((1-\mu)r\left(\partial_{\bar{t}}^k\Omega_\alpha^l\mathcal{G}F\right)^2+|\partial_{\bar{t}}^k\Omega_\alpha^lF|^2\right).
	\end{aligned}
	\end{equation}
	Then, to obtain all the lower order terms on the bulk of the left hand side of~\eqref{eq: sec: higher order G, proof eq 1.09} we sum in the Morawetz estimate~\eqref{eq: sec: higher order G, proof eq 1.1} of Proposition~\ref{prop: sec: higher order G estimate, prop 1}, at order $j-1$, and after appropriate Young's inequalities on the contribution of the $F$ error term we obtain
\begin{equation}\label{eq: sec: higher order G, proof eq 1.5}
\begin{aligned}
&	E^\prime_{\mathcal{G},j}[\psi](\tau_2)\\
&	+\int\int_{D(\tau_1,\tau_2)}\sum_{0\leq k+l= j-2}\sum_\alpha\left(\frac{1}{r(1-\mu)}(\partial_t\partial_{\bar{t}}^{k}\Omega_\alpha^{l}\mathcal{G}\psi)^2+\frac{1-\mu}{r}(\partial_r\partial_{\bar{t}}^{k}\Omega_\alpha^{l}\mathcal{G}\psi)^2+\frac{1}{r}\left|\slashed{\nabla}\partial_{\bar{t}}^{k}\Omega_\alpha^{l}\mathcal{G}\psi\right|^2\right)\\
&	\quad\quad\quad\quad\quad\quad\quad+\sum_{1\leq i_1+i_2+i_3\leq j-1}\sum_\alpha\left(\partial_{\bar{t}}^{i_1}\Omega_\alpha^{i_2}\partial_r^{i_3}\psi\right)^2\\
&	\leq C E^\prime_{\mathcal{G},j}[\psi](\tau_1)+C\int\int_{D(\tau_1,\tau_2)}(1-\mu)r\sum_{0\leq k+l\leq j-2}\sum_\alpha\left(\partial_{\bar{t}}^k\Omega_\alpha^l\mathcal{G}F\right)^2+\sum_{1\leq i_1+i_2+i_3\leq j-2}\sum_\alpha\left(\partial_{\bar{t}}^{i_1}\Omega_\alpha^{i_2}\partial_r^{i_3}F\right)^2\\
&	\qquad\qquad\qquad +C\int_{\phi_{\tau_2}(\Sigma)} \sum_{0\leq i_1+i_2+i_3\leq j-3}\sum_\alpha\left(\partial_{\bar{t}}^{i_1}\Omega_\alpha^{i_2}\partial_r^{i_3}F\right)^2
\end{aligned}
\end{equation}	
where we used the auxiliary energy 
\begin{equation}
	E^\prime_{\mathcal{G},j}[\psi](\tau)		=\int_{\{\bar{t}=\tau\}}\sum_{0\leq i_1+i_2= j-2}\sum_\alpha J_\mu^{\partial_{\bar{t}}}[\partial_{\bar{t}}^{i_1}\Omega_\alpha^{i_2}\mathcal{G}\psi]n^\mu+E_{j-1}[\psi](\tau).
\end{equation}
Note that in the energy estimate~\eqref{eq: sec: higher order G, proof eq 1.5} there is no degeneration, at the low order, on the photon sphere $r=3M$, because we repeated a Poincare type argument at top order, see Lemma \ref{lem: the Psi dot times psi dot}.

Note that we control all the desired higher order derivatives (at order $j$)  related to $\mathcal{G}$, on the left hand side of \eqref{eq: sec: higher order G, proof eq 1.5}, except for 
\begin{equation}
	\partial_r^{i_3}\mathcal{G}\psi,\qquad i_3=j-1
\end{equation}
with the appropriate degenerative weight. For that purpose, we return to the equation satisfied by $\mathcal{G}\psi$, see \eqref{eq: equation obeyed by Psi}, which reads 
\begin{equation}\label{eq: sec: higher order G, proof eq 3}
	\begin{aligned}
		&	g^{rr}\partial_r^2\mathcal{G}\psi=-\left(g^{\bar{t}\bar{t}}\partial_t^2\mathcal{G}\psi+2g^{r\bar{t}}\partial_t\partial_r\mathcal{G}\psi+\slashed{\Delta}\mathcal{G}\psi\right)+\frac{2\sqrt{1-9M^2\Lambda}}{\sqrt{1-\mu}\sqrt{1+\frac{6M}{r}}}\partial_{\bar{t}}\partial_r\psi+ E_1(r)\partial_{\bar{t}}\psi + E_2(r)\partial_r\psi+\mathcal{G}F,
	\end{aligned} 
\end{equation}
where recall $g^{rr}=1-\mu$. We differentiate \eqref{eq: sec: higher order G, proof eq 3} by
\begin{equation}
	\partial_r^{i_3-2},
\end{equation}
and then we square the result and note that there exists a constant $C(j,M,\Lambda)>0$ such that
\begin{equation}\label{eq: sec: higher order G, proof eq 5}
\begin{aligned}
(1-\mu)^2(\partial_r^{i_3}\mathcal{G}\psi)^2	&	\leq C(j,M,\Lambda) \left(\left(\partial_r^{i_3-2}\left(\frac{1}{\sqrt{1-\mu}}\right)\frac{2\sqrt{1-9M^2\Lambda}}{\sqrt{1+\frac{6M}{r}}} (\partial_r\partial_{\bar{t}}\psi)\right)^2+\dots+\left(\partial_r^{i_3-2}\mathcal{G}F\right)^2\right)\\
\end{aligned}
\end{equation}
where, in inequality~\eqref{eq: sec: higher order G, proof eq 5} the three terms displayed are from left to right the highest $\partial_r$ derivative term, the fastest degenerating $(1-\mu)$ term and the contribution of the inhomogeneity $F$. We rewrite~\eqref{eq: sec: higher order G, proof eq 5} as
\begin{equation}\label{eq: sec: higher order G, proof eq 5.1}
	\begin{aligned}
		(1-\mu)^2(\partial_r^{i_3}\mathcal{G}\psi)^2	&	\leq C(j,M,\Lambda)\left( \frac{1}{(1-\mu)^{2i_3-3}}\left(\partial_r\partial_{\bar{t}}\psi\right)^2+\dots+ \left(\partial_r^{i_3-2}\mathcal{G}F\right)^2\right)\\
	\end{aligned}
\end{equation}
and therefore, by multiplying~\eqref{eq: sec: higher order G, proof eq 5.1} with $\frac{(1-\mu)^{2i_3-3}}{r}$, so that no terms blow up at the roots of $1-\mu=0$, we obtain
\begin{equation}\label{eq: sec: higher order G, proof eq 5.2}
\begin{aligned}
\frac{(1-\mu)^{2i_3-1}}{r}(\partial_r^{i_3}\mathcal{G}\psi)^2	&	\leq C(j,M,\Lambda)\left(\frac{1}{r^3(1-\mu)} \left(\partial_{\bar{t}}\mathcal{G}\psi\right)^2+\dots+\frac{(1-\mu)^{2i_3-3}}{r}\left(\partial_r^{i_3-2}\mathcal{G}F\right)^2\right).
\end{aligned}
\end{equation}
For the $j=3$ case of inequality~\eqref{eq: sec: higher order G, proof eq 5.2} with all its terms displayed see already Remark~\ref{rem: sec: higher order G estimate, rem 1}. Therefore, by using the integrated inequality \eqref{eq: sec: higher order G, proof eq 1.5} and the pointwise estimate~\eqref{eq: sec: higher order G, proof eq 5.2} we obtain 
\begin{equation}\label{eq: sec: higher order G, proof eq 6}
\begin{aligned}
&	E^\prime_{\mathcal{G},j}[\psi](\tau_2) \\   &	\quad\quad +\int\int_{D(\tau_1,\tau_2)} \sum_{0\leq i_1+i_2= j-2}\sum_\alpha \frac{1}{r(1-\mu)}\left(\partial_{\bar{t}}^{1+i_1}\Omega_\alpha^{i_2}\mathcal{G}\psi\right)^2+\sum_{1\leq i_1+i_2\leq j-2}\sum_\alpha\frac{1}{r}\left|\slashed{\nabla}\partial_{\bar{t}}^{i_1}\Omega_\alpha^{i_2}\mathcal{G}\psi\right|^2\\
&	\quad\quad +\int\int_{D(\tau_1,\tau_2)}\sum_{1\leq i_1+i_2+i_3\leq j-1,i_3\geq 1}\sum_\alpha\frac{(1-\mu)^{2i_3-1}}{r}\left(\partial_{\bar{t}}^{i_1}\Omega_\alpha^{i_2}\partial_r^{i_3}\mathcal{G}\psi\right)^2+\sum_{1\leq i_1+i_2+i_3\leq j-1}\sum_\alpha\left(\partial_{\bar{t}}^{i_2}
\Omega_\alpha^{i_2}\partial_r^{i_3}\psi\right)^2\\
&   \quad \leq C  E^\prime_{\mathcal{G},j}[\psi](\tau_1)+C\int\int_{D(\tau_1,\tau_2)} (1-\mu)r\sum_{0\leq i_1+i_2\leq j-2}\sum_\alpha\left(\partial_{\bar{t}}^{i_1}\Omega_\alpha^{i_2}\mathcal{G}F\right)^2+\sum_{0\leq i_1+i_2+i_3\leq j-2}\sum_\alpha\left(\partial_{\bar{t}}^{i_1}\Omega_\alpha^{i_2}\partial_r^{i_3}F\right)^2.\\
&	\qquad\qquad\qquad\qquad+ C\int_{\phi_{\tau_2}(\Sigma)} \sum_{0\leq i_1+i_2+i_3\leq j-3}\sum_\alpha\left(\partial_{\bar{t}}^{i_1}\Omega_\alpha^{i_2}\partial_r^{i_3}F\right)^2
\end{aligned}
\end{equation}
for all $\tau_1\leq \tau_2$ and for all $j\geq 2$.

Now, we want to estimate from below $E^\prime_{\mathcal{G},j}[\psi](\tau_2)$ by $E_{\mathcal{G},j}[\psi](\tau_2)$, in the energy estimate \eqref{eq: sec: higher order G, proof eq 6}, for all orders $j\geq 3$ at the expense of producing hypersurface error terms. We note that there exists a constant $c(j,M,\Lambda)$ such that 
\begin{equation}\label{eq: sec: higher order G, proof eq 9}
\begin{aligned}
& c(j,M,\Lambda) \int_{\{\bar{t}=\tau\}} \sum_{i_1+i_2= j-2}\sum_\alpha\left( \frac{1}{r}(\partial_{\bar{t}}\partial_{\bar{t}}^{i_1}\Omega_\alpha^{i_2}\mathcal{G}\psi)^2+(1-\mu)r(\partial_r\partial_{\bar{t}}^{i_1}\Omega_\alpha^{i_2}\mathcal{G}\psi)^2+r\left|\slashed{\nabla}\partial_{\bar{t}}^{i_1}\Omega_\alpha^{i_2}\mathcal{G}\psi\right|^2 \right)\\
&	\qquad +c(j,M,\Lambda) E_{j-1}[\psi](\tau)\\
& \leq E^\prime_{\mathcal{G},j}[\psi](\tau).
\end{aligned}
\end{equation}
Then, we want to obtain the top order derivatives
\begin{equation}
\partial_r^{i_3}\mathcal{G}\psi,\qquad i_3=j-1,
\end{equation}
on the left hand side, with the appropriate degenerative weights of $(1-\mu)$. Therefore, we multiply the pointwise estimate~\eqref{eq: sec: higher order G, proof eq 5.1} with $(1-\mu)^{2i_3-3}$ and sum it to~\eqref{eq: sec: higher order G, proof eq 9} to obtain that there exist constants
\begin{equation}
	c(j,M,\Lambda),C(j,M,\Lambda)>0
\end{equation}
\begin{equation}\label{eq: sec: higher order G, proof eq 10}
	\begin{aligned}
		&	 c(j,M,\Lambda)\int_{\{\bar{t}=\tau\}} \sum_{0\leq i_1+i_2\leq j-2 }\sum_\alpha\mathcal{E}\left(\partial_{\bar{t}}^{i_1}\Omega_\alpha^{i_2}\mathcal{G}\psi,\partial_{\bar{t}}^{i_1}\Omega_\alpha^{i_2}\psi\right)+\sum_{1\leq i_1+i_2+i_3\leq j-1,i_3\geq 1}\sum_\alpha r(1-\mu)^{2i_3-1}\left(\partial_{\bar{t}}^{i_1}\Omega_\alpha^{i_2}\partial_r^{i_3}\mathcal{G}\psi\right)^2\\
		&	\qquad\qquad\qquad\qquad +\sum_{1\leq i_1+i_2+i_3\leq j-1,i_3\geq 1}\sum_\alpha r\left(\partial_{\bar{t}}^{i_1}\Omega_\alpha^{i_2}\partial_r^{i_3}\psi\right)^2\\
		&	\quad\leq  E^\prime_{\mathcal{G},j}[\psi](\tau)	\\
		&	\qquad + C(j,M,\Lambda)\int_{\{\bar{t}=\tau\}} \sum_{0\leq i_1+i_2+i_3\leq j-3}\sum_\alpha(1-\mu)^{2i_3+1}r\left(\partial_{\bar{t}}^{i_1}\Omega_\alpha^{i_2}\partial_r^{i_3}\mathcal{G}F\right)^2+\sum_{0\leq i_1+i_2+i_3\leq j-3}\sum_\alpha r\left(\partial_{\bar{t}}^{i_1}\Omega_\alpha^{i_2}\partial_r^{i_3}F\right)^2.
	\end{aligned}
\end{equation}
Note that the left hand side of~\eqref{eq: sec: higher order G, proof eq 10} is similar to $E_{\mathcal{G},j}[\psi](\tau)$.

Finally, by using the energy estimate \eqref{eq: sec: higher order G, proof eq 6} in conjunction with the pointwise estimate~\eqref{eq: sec: higher order G, proof eq 10} we conclude the result of the Corollary.
\end{proof}

\begin{remark}\label{rem: sec: higher order G estimate, rem 1}
For an example of the computation \eqref{eq: sec: higher order G, proof eq 5.2} note that at order $j=3$ we obtain  
\begin{equation}\label{eq: sec: higher order G, proof eq 4}
\begin{aligned}
\frac{(1-\mu)^3}{r}\left(\partial_r^2\mathcal{G}\psi\right)^2 &	\leq \frac{C(M,\Lambda)}{r}\Bigg(  (1-\mu)\left(g^{\bar{t}\bar{t}}\partial_t^2\mathcal{G}\psi+2g^{r\bar{t}}\partial_t\partial_r\mathcal{G}\psi+\slashed{\Delta}\mathcal{G}\psi\right)^2+(1-\mu)\frac{1}{(1-\mu)}(\partial_{\bar{t}}\partial_r\psi)^2\\
&	\quad\quad\quad\quad\quad\quad+(1-\mu)E_1^2(\partial_{\bar{t}}\psi)^2+(1-\mu)E_2^2(\partial_r\psi)^2+(1-\mu)(\mathcal{G}F)^2\Bigg)\\
& \leq \frac{C(M,\Lambda)}{r}\Bigg(  (1-\mu)\left(g^{\bar{t}\bar{t}}\partial_t^2\mathcal{G}\psi+2g^{r\bar{t}}\partial_t\partial_r\mathcal{G}\psi+\slashed{\Delta}\mathcal{G}\psi\right)^2+(\partial_{\bar{t}}\partial_r\psi)^2\\
&	\quad\quad\quad\quad\quad\quad+(1-\mu)E_1^2(r)(\partial_{\bar{t}}\psi)^2+(1-\mu)E_2^2(r)(\partial_r\psi)^2+(1-\mu)(\mathcal{G}F)^2\Bigg)\\
&\leq C(M,\Lambda)\Bigg(  \frac{1-\mu}{r^5}\left(\partial_t^2\mathcal{G}\psi\right)^2+\frac{1-\mu}{r}\left(\partial_t\partial_r\mathcal{G}\psi\right)^2+\frac{1-\mu}{r}\left(\slashed{\Delta}\mathcal{G}\psi\right)^2+\frac{1}{r^3(1-\mu)}(\partial_{\bar{t}}\mathcal{G}\psi)^2\\
&	\quad\quad\quad\quad\quad\quad+\frac{1-\mu}{r^3}(\partial_{\bar{t}}\psi)^2+\frac{1}{r^7}(\partial_r\psi)^2+\frac{1-\mu}{r}\left(\mathcal{G}F\right)^2\Bigg),
\end{aligned}
\end{equation}
where the constant $C(M,\Lambda)$ does not degenerate in the limit $\Lambda\rightarrow 0$.
\end{remark}

\section{The Schwarzschild case \texorpdfstring{$\Lambda=0$}{L}}\label{sec: thm in Schwarzschild}

We can also prove the equivalent of Theorem~\ref{thm: absorption of small terms} on the asymptotically flat Schwarzschild exterior. This result had been obtained previously by~\cite{gustav}. For completeness, we give a treatment here in our set up. 

We study the inhomogeneous wave equation 
\begin{equation}\label{eq: schwarzschild equation}
    \Box_{g_M} \psi =F.
\end{equation}

Before stating the theorem let us introduce some preliminary notions, specifically for the Schwarzschild case. 

First, we define regular hyperboloidal coordinates $(\bar{t},r,\theta,\phi)$ on which the metric takes the form 
\begin{equation}
    g_{M}=-\left(1-\frac{2M}{r}\right)(d\bar{t})^2-2\left(1-\frac{3M}{r}\right)\sqrt{1+\frac{6M}{r}} d\bar{t}dr +\frac{27M^2}{r^2} (dr)^2+r^2 d\sigma_{\mathbb{S}^2}.
\end{equation}
These correspond to the regular hyperboloidal coordinates of Section~\ref{subsec: regular hyperboloidal corodinates}, when we take $\Lambda=0$. Moreover, we can define the Schwarzschild coordinates $(t,r,\theta,\phi)$ similarly to the Schwarzschild--de~Sitter coordinates, see Section~\ref{sec: preliminaries}. 

We may attach null infinity 
\begin{equation}\label{eq: future null infinity}
    \mathcal{I}^+=\{(\bar{t},r=+ \infty,\theta,\phi)\in \mathbb{R}\times\{+\infty\}\times\mathbb{S}^2\},
\end{equation}
as a boundary in the obvious way. Formally, we consider its normal to be $n_{\mathcal{I}^+}=\frac{\partial}{\partial \bar{t}}$ and note 
\begin{equation}\label{eq: positive energy flux on null infinity}
J^T_\mu[\psi]n_{\mathcal{I}^+}^\mu=\left(\partial_{\bar{t}}\psi\right)^2.   
\end{equation}
 
Let
\begin{equation}
    \tilde{\Sigma}=\{\bar{t}=0\}.
\end{equation}
Then the leaves of the foliation
\begin{equation}\label{eq: spacelike hypersurface S}
\phi_\tau(\tilde{\Sigma})=\{\bar{t}=\tau\}    
\end{equation}
connect the event horizon $\mathcal{H}^+$ with future null infinity $\mathcal{I}^+$. Also, we define the spacetime domain
\begin{equation}\label{eq: spacetime domain S}
    \tilde{D}(\tau_1,\tau_2)=\{\tau_1\leq \bar{t}\leq \tau_2\}.
\end{equation}
The reader familiar with the Penrose diagram may want to consult Figure~\ref{fig: penrose schwarschild}. 
\begin{figure}[htbp]
    \begin{minipage}{.45\textwidth}
        \centering
        \includegraphics[scale=0.8]{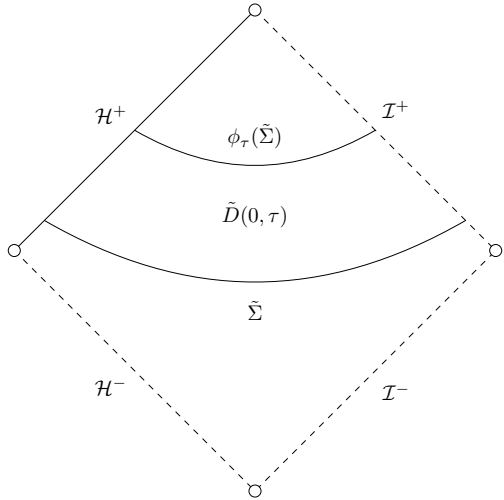}
        \caption{The foliation of the Schwarzschild exterior}
        \label{fig: penrose schwarschild}
    \end{minipage}
\end{figure}

Note that the volume form of \eqref{eq: spacetime domain S} is
\begin{equation}
    dg=r^2\sin{\theta}d\bar{t}d\theta d\phi
\end{equation}
and the volume form of \eqref{eq: spacelike hypersurface S} is
\begin{equation}
    dg_{\{\bar{t}=c\}}=3\sqrt{3}Mr\sin{\theta}dr d\theta d\phi.
\end{equation}
and the normal of the $\{\bar{t}=c\}$ hypersurface is 
\begin{equation}
	n_{\phi_{\tau}(\Sigma)}=\frac{\sqrt{27}M}{r}\frac{\partial}{\partial \bar{t}}+r\frac{\xi(r)}{\sqrt{27}M}\frac{\partial}{\partial r},
\end{equation}
where $\xi(r)=\left(1-\frac{3M}{r}\right)\sqrt{1+\frac{6M}{r}}$.

The vector field $\mathcal{G}$ here takes the form 
\begin{equation}\label{eq: new vector field Schwarzschild}
    \mathcal{G}=r\sqrt{1-\frac{2M}{r}}\frac{\partial}{\partial r}=\frac{r}{\sqrt{1-\mu}}\frac{\partial}{\partial r^\star}+\frac{r}{\sqrt{1-\mu}}\left(1-\frac{3M}{r}\right)\sqrt{1+\frac{6M}{r}}\frac{\partial}{\partial t}
\end{equation}
and the energy density, 
\begin{equation}\label{eq: new energy, S}
\begin{aligned}
\tilde{\mathcal{E}}(\mathcal{G}\psi,\psi)\:\dot{=}\: \mathbb{T}(\partial_{\bar{t}},n)[\mathcal{G}\psi]+\mathbb{T}(N,n)[\psi],
\end{aligned}
\end{equation}
is now defined with $\mathcal{G}$ as in \eqref{eq: new vector field Schwarzschild}, where $N$ is a time translation invariant strictly timelike vector field, see the Lecture notes~\cite{DR5} that away from the horizon $\mathcal{H}^+$ is equal to $\partial_{\bar{t}}$. We also obtain 
\begin{equation}
    \begin{aligned}
        \tilde{\mathcal{E}}(\mathcal{G}\psi,\psi) &\sim_{M} \frac{1}{r}(\partial_{\bar{t}}\mathcal{G}\psi)^2+r\left(1-\frac{2M}{r}\right)(\partial_r\mathcal{G}\psi)^2+r|\slashed{\nabla}\mathcal{G}\psi|^2+\frac{1}{r}(\partial_{\bar{t}}\psi)^2+r(\partial_r\psi)^2+r|\slashed{\nabla}\psi|^2 \\
        &\sim_{M} r\left(\frac{1}{r^2}(\partial_{\bar{t}}\mathcal{G}\psi)^2+\left(1-\frac{2M}{r}\right)(\partial_r\mathcal{G}\psi)^2+|\slashed{\nabla}\mathcal{G}\psi|^2+\frac{1}{r^2}(\partial_{\bar{t}}\psi)^2+(\partial_r\psi)^2+|\slashed{\nabla}\psi|^2\right).
    \end{aligned}
\end{equation}

\begin{remark}
Note that the vector field~\eqref{eq: new vector field Schwarzschild} coincides with the vector field already described by Holzegel--Kauffman~\cite{gustav}.
\end{remark}

Note that the divergence theorem of equation \eqref{eq: divergence Theorem}, holds with $\mathcal{I}^+$ in the place of $\bar{\mathcal{\mathcal{H}}}^+$. The energy flux at null infinity $\mathcal{I}^+$, and the event horizon $\mathcal{H}^+$, are nonnegative and we will may drop them from our estimates, see equation \eqref{eq: positive energy flux on null infinity}.

The Morawetz estimate we will need was proved in~\cite{DR1} by Dafermos--Rodnianski. 
\begin{theorem}[Theorem 1.1 in~\cite{DR1}]\label{thm: morawetz in schwarzschild}
Let $\psi$ satisfy equation~\eqref{eq: schwarzschild equation} on a fixed Schwarzschild background. Then, for initial data vanishing at infinity, for all $0<\eta\leq 1$ there exists a constant $C=C(M,\eta)$ such that
\begin{equation}\label{eq: morawetz in schwarzschild}
    \begin{aligned}
        \int\int_{\tilde{D}(0,\tau)}r^{-1}\left(1-\frac{3M}{r}\right)^2\left(r^{-\eta}(\partial_{\bar{t}}\psi)^2+\left|\slashed{\nabla}\psi\right|^2\right)&+r^{-1-\eta}\left(\partial_{r}\psi\right)^2+r^{-3-\eta}\psi^2\\
        &   \leq C\int_{\tilde{\Sigma}} J^N_\mu[\psi]n^\mu+C\int\int_{\tilde{D}(0,\tau)}|\partial_{\bar{t}}\psi\cdot F|+r^{1+\eta}|F|^2
    \end{aligned}
\end{equation}
and 
\begin{equation}\label{eq: boundedness in schwarzschild}
	\begin{aligned}
		\int_{\phi_\tau(\tilde{\Sigma})}J^N_\mu[\psi]n^\mu\leq C \int_{\tilde{\Sigma}}J^N_\mu[\psi]n^\mu+C\int\int_{\tilde{D}(0,\tau)}|\partial_{\bar{t}}\psi \cdot F|+r^{1+\eta}|F|^2
	\end{aligned}
\end{equation}
where $N$ is the redshift vector field of Dafermos--Rodnianski, see the Lecture notes~\cite{DR5}. 
\end{theorem}

Recalling Remark~\ref{rem: keeping track of the powers of r}, we repeat similar arguments to the ones used in the previous Section \ref{sec: proof of main theorem 1}, up to estimate \eqref{eq: last inequality before the proof of theorem 1}. Note, however, that we use the commutation vector field \eqref{eq: new vector field Schwarzschild} and the energy density \eqref{eq: new energy} with $\Lambda=0$. Now, adding the Morawetz estimate of Theorem \ref{thm: morawetz in schwarzschild}, noting the weights and absorbing relevant terms, we conclude the following.
\begin{theorem}\label{thm: Schwarzschild}
There exists a constant $C(M,\eta)>0$, such that if $\psi$ satisfies
\begin{equation*}
    \Box_{g_M}\psi=F
\end{equation*}
on the spacetime domain $\tilde{D}(0,\tau)$, the following holds
\begin{equation}
    \begin{aligned}
        \int_{\phi_\tau(\tilde{\Sigma})}\tilde{\mathcal{E}}(\mathcal{G}\psi,\psi)+\int\int_{\tilde{D}(0,\tau)} &  \frac{1}{r(1-\mu)}(\partial_{\bar{t}}\mathcal{G}\psi)^2+\frac{1-\mu}{r}(\partial_r\mathcal{G}\psi)^2+\frac{1}{r}|\slashed{\nabla}\mathcal{G}\psi|^2+\frac{1}{r^{1+\eta}}\left((\partial_{\bar{t}}\psi)^2+(\partial_r\psi)^2+|\slashed{\nabla}\psi|^2\right)\\
        &   \quad \leq C\int_{\tilde{\Sigma}}\tilde{\mathcal{E}}(\mathcal{G}\psi,\psi)+C\int\int_{\tilde{D}(0,\tau)} r|1-\mu|\left|\mathcal{G}F\right|^2+r^{1+\eta}|F|^2
    \end{aligned}
\end{equation}
for any $\eta\in(0,1]$.
\end{theorem}
\begin{proof}
By repeating the arguments of Section \ref{sec: proof of main theorem 1}, and by keeping track of the weights in $r$, we conclude the following estimate
\begin{equation}\label{eq: proof thm: Schwarzschild, eq 1}
	\begin{aligned}
		&	\int_{\phi_\tau(\tilde{\Sigma})}J^{\partial_{\bar{t}}}_\mu[\mathcal{G}\psi]n^\mu+\int\int_{\tilde{D}(0,\tau)}   \frac{1}{r(1-\mu)}(\partial_{\bar{t}}\mathcal{G}\psi)^2\\
		&	\quad \lesssim \int_{\tilde{\Sigma}}J^{\partial_{\bar{t}}}_\mu[\mathcal{G}\psi]n^\mu+\int\int_{\tilde{D}(0,\tau)}\partial_{\bar{t}}\mathcal{G}\psi \cdot \mathcal{G}F+\partial_{\bar{t}}\mathcal{G}\psi\left(E_2(r)\partial_t\psi+E_1(r)\partial_r\psi\right),
	\end{aligned}
\end{equation}
where for $E_1(r),E_2(r)$ see Proposition \ref{eq: equation obeyed by Psi}. Then, to obtain all of the derivatives of $\mathcal{G}\psi$ on the left hand side of~\eqref{eq: proof thm: Schwarzschild, eq 1} we use Lemma~\ref{lem: control of Psi prime}, by keeping the weights in $r$ and by an additional $F$ error term to obtain 
\begin{equation}\label{eq: proof thm: Schwarzschild, eq 1.1}
\begin{aligned}
&	\int_{\phi_\tau(\tilde{\Sigma})}J^{\partial_{\bar{t}}}_\mu[\mathcal{G}\psi]n^\mu+\int\int_{\tilde{D}(0,\tau)}   \frac{1}{r(1-\mu)}(\partial_{\bar{t}}\mathcal{G}\psi)^2+\frac{1-\mu}{r}\left(\partial_r\mathcal{G}\psi\right)^2+\frac{1}{r}|\slashed{\nabla}\mathcal{G}\psi|^2 \\
&	\quad \lesssim \int_{\tilde{\Sigma}}J^{\partial_{\bar{t}}}_\mu[\mathcal{G}\psi]n^\mu\\
&	\quad\quad + \int\int_{\tilde{D}(\tau_1,\tau_2)} r^{-2}(\partial_r \psi)^2 + r^{-2} (\partial_t\psi)^2 +\int_{\phi_{\tau_2}(\tilde{\Sigma})} r(\partial_r\psi)^2 +\int_{\phi_{\tau_1}(\tilde{\Sigma})} r(\partial_r\psi)^2 \\
&\quad\quad +\int\int_{\tilde{D}(0,\tau)}\partial_{\bar{t}}\mathcal{G}\psi \cdot \mathcal{G}F+\partial_{\bar{t}}\mathcal{G}\psi\left(E_2(r)\partial_t\psi+E_1(r)\partial_r\psi\right)+\left|\frac{1}{r}\mathcal{G}\psi\cdot F\right|.
\end{aligned}
\end{equation}
To control all the first order terms on the right hand side of~\eqref{eq: proof thm: Schwarzschild, eq 1.1} we use the boundedness estimate and the Morawetz estimate of Theorem~\ref{thm: morawetz in schwarzschild} and obtain  
\begin{equation}\label{eq: proof thm: Schwarzschild, eq 2}
\begin{aligned}
&	\int_{\phi_\tau(\tilde{\Sigma})}\tilde{\mathcal{E}}(\mathcal{G}\psi,\psi)\\
&	+\int\int_{\tilde{D}(0,\tau)}  \frac{1}{r(1-\mu)}(\partial_{\bar{t}}\mathcal{G}\psi)^2+\frac{1-\mu}{r}(\partial_r\mathcal{G}\psi)^2+\frac{1}{r}|\slashed{\nabla}\mathcal{G}\psi|^2+\frac{1}{r^{1+\eta}}\left((\partial_{\bar{t}}\psi)^2+(\partial_r\psi)^2+|\slashed{\nabla}\psi|^2\right)\\
&	\quad \lesssim \int_{\tilde{\Sigma}}\tilde{\mathcal{E}}(\mathcal{G}\psi,\psi)+\int\int_{\tilde{D}(0,\tau)}|\partial_{\bar{t}}\mathcal{G}\psi \cdot \mathcal{G}F|+|\partial_{\bar{t}}\psi \cdot F|+r^{1+\eta}|F|^2,
\end{aligned}
\end{equation}
for any $\eta\in(0,1]$, where we used Young's inequalities in the lower order terms on the right hand side of \eqref{eq: proof thm: Schwarzschild, eq 1}. Note that no degeneration is present on the photon sphere $r=3M$ since we repeated the Poincare type inequality of Lemma~\ref{lem: the Psi dot times psi dot}. Finally, by using the appropriate Young's inequalities on the right hand side of \eqref{eq: proof thm: Schwarzschild, eq 2}, we conclude the result. 
\end{proof}

Moreover, we have the Corollary.

\begin{corollary}\label{cor: Schwarzschild}
Under the assumptions of the above theorem, and the additional assumptions that, for the vector field $a=a^j\partial_j$ and the scalar function $b$ the following
\begin{equation}
a^{\bar{t}},a^r,g^{\theta\theta}(a^\theta)^2+g^{\phi\phi}(a^\phi)^2,b(r)    
\end{equation}
are smooth functions on $\tilde{D}(0,\infty)$, with 
\begin{equation}
    a^{\bar{t}}(r,\cdot),a^r(r,\cdot),\left(g^{\theta\theta}(a^\theta)^2+g^{\phi\phi}(a^\phi)^2\right)(r,\cdot),|[\mathcal{G}, a]^{\bar{t}}(r,\cdot)|\leq\frac{C}{r^2}   
\end{equation}
and 
\begin{equation}
    b(r,\cdot), |\mathcal{G}b|(r,\cdot)\leq\frac{C}{r^3},
\end{equation}
then we have the following. 

There exists a constant $C=C(M,\eta)$ such that for solutions of
\begin{equation}\label{eq: small error equation S}
    \Box_{g_{M}}\psi=\epsilon a^j\partial_j\psi+\epsilon b(r)\psi,
\end{equation}
and for $\epsilon$ sufficiently small, we have 
\begin{equation}\label{eq: result of Corollary 2}
	\begin{aligned}
	    &	\int_{\phi_\tau(\tilde{\Sigma})}\tilde{\mathcal{E}}(\mathcal{G}\psi,\psi)\\
	    &	\quad +\int\int_{\tilde{D}(0,\tau)}  \frac{1}{r(1-\mu)}(\partial_{\bar{t}}\mathcal{G}\psi)^2+\frac{1-\mu}{r}(\partial_r\mathcal{G}\psi)^2+\frac{1}{r}|\slashed{\nabla}\mathcal{G}\psi|^2+\frac{1}{r^{1+\eta}}\left((\partial_{\bar{t}}\psi)^2+(\partial_r\psi)^2+|\slashed{\nabla}\psi|^2\right)\\
	    &	\leq C\int_{\tilde{\Sigma}}\tilde{\mathcal{E}}(\mathcal{G}\psi,\psi),
	\end{aligned}
\end{equation}
and
\begin{equation}\label{eq: result of Corollary 2, 2}
	\begin{aligned}
		&	\int_{\phi_\tau(\tilde{\Sigma})}\tilde{\mathcal{E}}(\mathcal{G}\psi,\psi)\\
		&	\quad +\int_0^\tau d\tau\int_{\phi_\tau(\tilde{\Sigma})}\frac{1}{(1-\mu)}(\partial_{\bar{t}}\mathcal{G}\psi)^2+(1-\mu)(\partial_r\mathcal{G}\psi)^2+|\slashed{\nabla}\mathcal{G}\psi|^2+\frac{1}{r^{\eta}}\left((\partial_{\bar{t}}\psi)^2+(\partial_r\psi)^2+|\slashed{\nabla}\psi|^2\right)\\
		&	\leq C\int_{\tilde{\Sigma}}\tilde{\mathcal{E}}(\mathcal{G}\psi,\psi),
	\end{aligned}
\end{equation}
since by the relation of the volume forms, see Section \ref{subsec: volume forms}, we have
\begin{equation}\label{eq: result of Corollary 2, 3}
    \int_{0}^\tau d\tau\int_{\phi_\tau(\Sigma)} \left(\cdot\right)\sim\int\int_{D(0,\tau)}\frac{1}{r}\left(\cdot\right),
\end{equation}
where the constants in the above similarity depend only on the black hole mass $M$.
\end{corollary}

\begin{remark}
The bulk term, on the left hand side of inequality \eqref{eq: result of Corollary 2, 2}, and the hypersurface term, on the right hand side of inequality \eqref{eq: result of Corollary 2, 2}, have different weights in $r$. Specifically, some of the terms of the right hand side, have larger weights in $r$. It is because of this reason that one does not obtain exponential decay for $\psi$, on a Schwarzschild exterior. 
\end{remark}

\begin{remark}
The reason we can include the zero'th order term, with component $b(r)$, in equation \eqref{eq: small error equation S}, as opposed to equation~\eqref{eq: small error equation SdS} on the Schwarzschild--de~Sitter case, is the lack of the $\psi_\infty$ term in~\eqref{eq: morawetz in schwarzschild} for initial data vanishing at infinity. 
\end{remark}
\begin{remark}
Our Corollary~\ref{cor: Schwarzschild}, without the $\frac{1}{r(1-\mu)}(\partial_{\bar{t}}\mathcal{G}\psi)^2$ term and slightly different weights of $r$ on the left hand side, coincides with Theorem 4.1 of~\cite{gustav}. 
\end{remark}

Note that we have the following higher order Theorem

\begin{theorem}\label{cor: Schwarzschild, cor 2}
Under the assumption of Theorem \ref{thm: Schwarzschild}, for all $j\geq 3$ there exists a positive constant $C=C(j,M,\eta)>0$ such that 

\begin{equation}\label{eq: cor: Schwarzschild, cor 2, eq 1}
\begin{aligned}
&	\tilde{E}_{\mathcal{G},j}[\psi](\tau_2)\\
&\quad +\int\int_{\tilde{D}(\tau_1,\tau_2)}  \sum_{0\leq i_1+i_2= j-2}\sum_\alpha \frac{1}{r(1-\mu)}\left(\partial_{\bar{t}}\partial_{\bar{t}}^{i_1}\Omega_\alpha^{i_2}\mathcal{G}\psi\right)^2+\frac{1}{r}\left|\slashed{\nabla}\partial_{\bar{t}}^{i_1}\Omega_\alpha^{i_2}\mathcal{G}\psi\right|^2 \\
&	\quad\quad\quad\quad\quad\quad+\sum_{1\leq i_1+i_2+i_3=j-1}\sum_\alpha \frac{(1-\mu)^{2i_3-1}}{r}\left(\partial_{\bar{t}}^{i_1}\Omega_\alpha^{i_2}\partial_r^{i_3}\mathcal{G}\psi\right)^2\\
&	\quad\quad\quad\quad\quad\quad+\sum_{1\leq i_1+i_2+i_3\leq j-1,i_2<j-1}\sum_\alpha \frac{1}{r^{1+\eta}}\left(\partial_{\bar{t}}^{i_1}\Omega_\alpha^{i_2}\partial_r^{i_3}\psi\right)^2+ \frac{1}{r}\left|\slashed{\nabla}\partial_{\bar{t}}^{i_1}\Omega_\alpha^{i_2}\partial_r^{i_3}\psi\right|^2\\
&   \quad \leq C  \tilde{E}_{\mathcal{G},j}[\psi](\tau_1)+C\int\int_{D(\tau_1,\tau_2)} (1-\mu)r\sum_{0\leq i_1+i_2\leq j-2}\sum_\alpha\left(\partial_{\bar{t}}^{i_1}\Omega_\alpha^{i_2}\mathcal{G}F\right)^2+\sum_{0\leq i_1+i_2+i_3\leq j-2}\sum_\alpha  r^{1+\eta}\left(\partial_{\bar{t}}^{i_1}\Omega_\alpha^{i_2}\partial_r^{i_3}F\right)^2\\
&	\quad\quad\quad +C \int_{\{\bar{t}=\tau_2\}} \sum_{0\leq i_1+i_2\leq j-3}\sum_\alpha \tilde{\mathcal{E}}\left(\partial_{\bar{t}}^{i_1}\Omega_\alpha^{i_2}\mathcal{G}F,\partial_{\bar{t}}^{i_1}\Omega_\alpha^{i_2}F\right)+\sum_{1\leq i_1+i_2+i_3\leq j-3}\sum_\alpha r(1-\mu)^{2i_3+1}\left(\partial_{\bar{t}}^{i_1}\Omega_\alpha^{i_2}\partial_r^{i_3}\mathcal{G}F\right)^2\\
&	\qquad\qquad\qquad +\sum_{0\leq i_1+i_2+i_3\leq j-3}\sum_\alpha r\left(\partial_{\bar{t}}^{i_1}\Omega_\alpha^{i_2}\partial_r^{i_3}F\right)^2
\end{aligned}
\end{equation}
for all $0\leq \tau_1\leq \tau_2$, where 
\begin{equation}\label{eq: cor: Schwarzschild, cor 2, eq 2}
\begin{aligned}
\tilde{E}_{\mathcal{G},j}[\psi](\tau)	& = \int_{\{\bar{t}=\tau\}} \sum_{0\leq i_1+i_2\leq j-2 }\sum_\alpha\tilde{\mathcal{E}}\left(\partial_{\bar{t}}^{i_1}\Omega_\alpha^{i_2}\mathcal{G}\psi,\partial_{\bar{t}}^{i_1}\Omega_\alpha^{i_2}\psi\right)+\sum_{1\leq i_1+i_2+i_3\leq j-1,i_3\geq 1}\sum_\alpha r(1-\mu)^{2i_3-1}\left(\partial_{\bar{t}}^{i_1}\Omega_\alpha^{i_2}\partial_r^{i_3}\mathcal{G}\psi\right)^2\\
&	\qquad\qquad +\sum_{1\leq i_1+i_2+i_3\leq j-1,i_3\geq 1}\sum_\alpha r\left(\partial_{\bar{t}}^{i_1}\Omega_\alpha^{i_2}\partial_r^{i_3}\psi\right)^2
\end{aligned}
\end{equation}
where $N$ is the redshift vector field of Dafermos--Rodnianski, see the Lecture notes~\cite{DR5}. The $j=2$ case is the same without the hypersurface error terms on the right hand side of \eqref{eq: cor: Schwarzschild, cor 2, eq 1}. For the volume forms of the spacetime domains and of the hypersurface terms see Section \ref{subsec: volume forms}. 
\end{theorem}
\begin{proof}
	The proof of this Theorem is essentially the same as Theorem~\ref{thm: higher order G estimate} of the Schwarzschild--de~Sitter case. Note that in the proof of Theorem~\ref{thm: higher order G estimate} we kept some `key' weights in $r$ so that the proof can be read for the asymptotically flat case. 	
\end{proof}

\begin{remark}
In order to write inequality \eqref{eq: cor: Schwarzschild, cor 2, eq 1} in its most compact form we had to slightly differ its presentation from the relevant inequality of Theorem~\ref{thm: higher order G estimate}, since the weights in $r$ are now significant. 
\end{remark}

\appendix

\section{The Christoffel symbols of the metric \texorpdfstring{\eqref{eq: regular metric de sitter}}{metric}}\label{sec: christoffel symbols}

The non-zero Christoffel symbols, of the metric \eqref{eq: regular metric de sitter}, are 
\begin{equation}
    \begin{aligned}
        &   \Gamma^{\bar{t}}\:_{\bar{t}\bar{t}}=\frac{\sqrt{1+\frac{6M}{r}}(r-3M)(-3M+r^2\Lambda)}{3r^3\sqrt{1-9M^2\Lambda}},\quad \Gamma^{\bar{t}}\:_{\bar{t}r}=\frac{9M^2(3M-r^3\Lambda)}{r^4(1-9M^2\Lambda)} \\
        &   \Gamma^{\bar{t}}\:_{rr}=\frac{27M^2(9M^2+3Mr+r^2)}{\sqrt{1+\frac{6M}{r}}r^5\left(1-9M^2\Lambda\right)^{3/2}} ,\quad \Gamma^{\bar{t}}\:_{\theta \theta}=\frac{\sqrt{1+\frac{6M}{r}}(r-3M)}{\sqrt{1-9M^2\Lambda}},\\
        &   \Gamma^{\bar{t}}\:_{\phi\phi}=\frac{\sqrt{1+\frac{6M}{r}}(r-3M)\sin{\theta}}{\sqrt{1-9M^2\Lambda}}, \quad \Gamma^r\:_{\bar{t}\bar{t}}=\frac{(-3M+r^3\Lambda)(6M-3r+r^3\Lambda)}{9 r^3},\\
        &   \Gamma^r\:_{\bar{t}r}=-\frac{\sqrt{1+\frac{6M}{r}}(r-3M)(-3M+r^3\Lambda)}{3r^3\sqrt{1-9M^2\Lambda}},\quad \Gamma^{r}_{rr}=-\frac{9M^2(3M-r^3\Lambda)}{r^4(1-9M^2\Lambda)}\\
        &   \Gamma^r\: _{\theta \theta}=-r\left(1-\frac{2M}{r}-\frac{\Lambda}{3}r^2\right),\quad \Gamma^r \:_{\phi\phi}=-r\left(1-\frac{2M}{r}-\frac{\Lambda}{3}r^2\right)\sin{\theta} \\
        &   \Gamma^{\theta}\:_{r\theta}=\frac{1}{r},\quad \Gamma^{\theta}\:_{\phi\phi}=-\frac{\cos{\theta}}{2},\quad \Gamma^{\phi}\:_{r\phi}=\frac{1}{r},\quad \Gamma^{\phi}\:_{\theta \phi}=\frac{\cot{\theta}}{2}.
    \end{aligned}
\end{equation}

%%%%%%%%%%%%%%%%%%%%%%%%%%%%%%%%%%%%%%%%%%%%%%%%%%%%%%%%%%%%%
\bibliographystyle{plain}
\bibliography{MyBibliography}

\end{document}